\documentclass[12pt]{article}
\usepackage{amsmath, amssymb, amsthm}
\usepackage{graphicx, url}
\usepackage{rotating}
\usepackage{multicol}
\usepackage[small,it]{caption}
\usepackage{subfig}
\usepackage{fullpage}
\usepackage{setspace}
\usepackage{booktabs}
\usepackage{epigraph}
\usepackage[all]{xy}
\usepackage{placeins}
\usepackage{bm}
\usepackage[page]{appendix}

\usepackage{verbatim}
\usepackage{natbib}
\usepackage{ifxetex}
\ifxetex
  \usepackage{unicode-math}
\else
\usepackage{mathpazo}
  \usepackage{tgtermes}
\fi

\usepackage{pgf, tikz}
\usetikzlibrary{positioning}
\usetikzlibrary{arrows}
\newcommand{\E}{\mathbb{E}}

\renewcommand{\P}{\mathbb{P}}

\newtheorem{proposition}{Proposition}
\newtheorem{assumption}{Assumption}
\newtheorem{lemma}{Lemma}
\newcommand{\indep}{\mbox{$\perp\!\!\!\perp$}}

\definecolor{gray}{rgb}{0.459,0.438,0.471}
\definecolor{crimson}{rgb}{0.6,0,0}

\usepackage[plainpages=false,
            pdfpagelabels,
            bookmarksnumbered,
            hidelinks,
            pdftitle={},
            pdfauthor={},
            pdfkeywords={}]{hyperref}

\newcommand{\blind}{0}
\newcommand{\submission}{0}

\if1\submission
\usepackage{xr}
\fi

\newcommand{\tit}{Priming bias versus post-treatment bias in
  experimental designs}

\if0\blind

{\title{\bf \tit\thanks{Thanks to Dean Knox, Fredrick S{\"a}vje, and two anonymous reviewers from the Alexander and Diviya Magaro Peer Pre-Review at Harvard's Institute for Quantitative Social Science for helpful comments and feedback. Open source software to implement the method of this paper are included in the \texttt{prepost} R package. Imai acknowledges financial support from the
      National Science Foundation (SES--0752050).}}

  \author{Matthew Blackwell\thanks{Associate Professor, Department of Government, Harvard University. 1737 Cambridge Street,
      Institute for Quantitative Social Science, Cambridge MA 02138. 
      Email: \href{mailto:mblackwell@gov.harvard.edu}{mblackwell@gov.harvard.edu}, 
      URL: \href{https://mattblackwell.org}{https://mattblackwell.org}} \quad\quad 
      Jacob R. Brown\thanks{Assistant Professor, Department of Political Science, Boston University. 232 Bay State Road, Boston, MA 02215. 
      Email: \href{mailto:jbrown13@bu.edu}{jbrown13@bu.edu}, 
      URL: \href{https://jacobrbrown.com}{https://jacobrbrown.com}} \quad\quad 
      Sophie Hill\thanks{PhD Student, Department of Government, Harvard University. 1737 Cambridge Street, Cambridge MA 02138. 
      Email: \href{mailto:sophie_hill@g.harvard.edu}{sophie\_hill@g.harvard.edu}, 
      URL: \href{https://www.sophie-e-hill.com/}{sophie-e-hill.com}} \quad\quad \\
     Kosuke Imai\thanks{Professor, Department of Government and Department of
      Statistics, Harvard University.  1737 Cambridge Street,
      Institute for Quantitative Social Science, Cambridge MA 02138.
      Email: \href{mailto:imai@harvard.edu}{imai@harvard.edu} URL:
      \href{https://imai.fas.harvard.edu}{https://imai.fas.harvard.edu}}
      \quad\quad Teppei Yamamoto\thanks{Professor, Faculty of Political Science and Economics, Waseda University, Tokyo, Japan. Email:
      \href{mailto:tyamam@waseda.jp}{tyamam@waseda.jp}, URL:
      \href{http://web.mit.edu/teppei/www}{http://web.mit.edu/teppei/www}}
      }
\date{\today \\ Word Count: 7,226} 
\fi
\if1\blind
\title{\bf \tit}
\date{\today \\ Word Count: 7,226}
\fi

\begin{document}

\maketitle

\vspace{-1em}
\begin{abstract}
  Conditioning on variables affected by treatment can induce
  post-treatment bias when estimating causal effects. Although this
  suggests that researchers should measure potential moderators before
  administering the treatment in an experiment, doing so may also bias
  causal effect estimation if the covariate measurement primes
  respondents to react differently to the treatment. This paper
  formally analyzes this trade-off between post-treatment and priming
  biases in three experimental designs that vary when moderators are
  measured: pre-treatment, post-treatment, or a randomized choice
  between the two. We derive nonparametric bounds for interactions
  between the treatment and the moderator under each design and show
  how to use substantive assumptions to narrow these bounds. These
  bounds allow researchers to assess the sensitivity of their
  empirical findings to priming and post-treatment bias.  We then apply the
  proposed methodology to a survey experiment on electoral messaging.

\end{abstract}

\baselineskip=1.57\baselineskip

\newpage
\section{Introduction}
\label{sec:intro}

Ascertaining heterogeneous treatment effects is an integral part of
many survey experiments. Researchers are often interested in how
treatment effects vary across respondents with different
characteristics. For example, we may be interested both in how
implicit versus explicit racial cues affect support for a particular
policy but also in how those effects differ by levels of racial
resentment \citep{vale:hutc:whit:02}. Or we may
want to know whether the effect of land-based electoral appeals might depend on the voters' sense of land security
\citep{horowitz2020can}.  These questions of effect
heterogeneity allow researchers to explore potential causal mechanisms
and design more targeted and effective future treatments.

To examine such treatment effect heterogeneity, we must measure the
relevant covariates, such as racial resentment or land security in the
aforementioned examples, at some point during the survey experiment.
The question of \emph{when} we measure these moderators, however, is a
source of methodological debate. On the one hand, a long tradition in
political science has recognized the potential \emph{priming bias} of
a {\it pre-test design}, where covariates are measured prior to
treatment \citep[e.g.,][]{tran:07, morr:carr:fox:08, klar:13,
  klar:leep:robi:19, schi:mont:pesk:22}. For example, asking a
respondent about their party identification might lead them to
evaluate the treatment in a more partisan or political light,
resulting in biased causal effect estimates. Several studies have
documented priming effects from a range of different covariates
\citep[see][for a review]{klar:leep:robi:19}.  Some find that certain
priming effects can last for weeks
\citep{chon:druc:10}.
 
On the other hand, the practice of measuring moderators after
treatment, what we call a {\it post-test design}, has come under
scrutiny due to the possibility for \emph{post-treatment bias}
\citep{Rosenbaum84, AchBlaSen16, MonNyhTor18}. In particular, if
covariates are affected by the treatment, then conditioning on those
covariates---as required when assessing effect heterogeneity---can
bias the estimation of conditional average treatment effect and
any interactions that compare such effects. In our empirical
application, the key moderating variable of land insecurity is a
subjective, perceived measure and thus potentially affected by the
framing of a political appeal around land
rights.  
Though treatment is unlikely to affect measurements of many moderators
like basic demographics, researchers investigating manipulable
moderators face a dilemma about when to measure these covariates when
designing an experiment.

In this paper, we apply the nonparametric identification and bounding approach
of \cite{mans:95} and \cite{balk:pear:97} to formally analyze the trade-off
between priming and post-treatment biases under different experimental designs,
and propose principled ways to analyze data \citep[see ][for a similar analysis
of measurement error]{imai:yama:10}.

We begin by deriving nonparametric bounds to show that neither the pre-test nor
post-test design provides much information about conditional average treatment
effects or interactions without additional assumptions. Next, we show how three
potentially plausible assumptions can narrow the bounds. The first is
\emph{priming monotoncity}, which assumes that whether the moderator is measure
pre- or post-treatment only affects the outcome in one direction. The second is
\emph{moderator monotonicity}, which assumes that measuring the covariates after
treatment can move the value of that moderator only in one direction. The third
assumption is \emph{stable moderator under control}, whereby the covariate under
the control condition cannot be affected by the timing of treatment. None of
these assumptions can point identify the interaction between the treatment and a
moderator, but they can substantially narrow the bounds and sometimes be
informative about the sign of such an interaction.

To further sharpen our inference, we generalize the two standard experimental
designs and consider a {\it randomized placement design}, where the experimenter
randomly assigns respondents to either the pre-test or post-test design. We
demonstrate how to estimate nonparametric bounds in this context and how to
incorporate assumptions that connect the pre-test and post-test arms to further
narrow the bounds. In our supplemental materials, we also develop a parametric
Bayesian approach to incorporate pre-treatment covariates in the analysis to
sharpen our inferences and quantify estimation uncertainty.

We also derive sensitivity analysis procedures for all three designs. In these
analyses, we vary the proportion of respondents whose outcomes or moderators are
affected by when the moderator is measured and assess how the bounds change as a
function of this sensitivity parameter. This procedure allows researchers and
readers to gauge the credibility of an estimated interaction in light of a more
nuanced set of assumptions, rather than the blunt instruments of the
monotonicity and stability.

Several recent studies have empirically explored the trade-off between priming
and post-treatment bias. \cite{AlbJes22} find that the effect of moderator
placement has little effect on the estimated interaction in a question-wording
experiment. Furthermore, they find no evidence for an average effect of
treatment on the moderator, potentially reducing concerns about post-treatment
bias. \cite{SheCli23} compared several experiments when the moderators were
measured just prior to treatment or in a prior survey wave. The authors found
that estimated effects and interactions were similar across these conditions and
concluded that priming bias may not be a widespread concern for experimental
studies in political science.

Both of these papers present compelling evidence for the specific
experiments conducted in their empirical assessments, but as
\cite{SheCli23} warn us, ``we should be cautious in generalizing
[these] findings to the wide variety of studies run by political
scientists.''  Our approach, on the other hand, provides a general
methodological toolkit that can be applied to any experimental design
and allows researchers to include substantive assumptions to tailor
the framework to their applications.

The rest of the paper proceeds as follows. We first introduce a
motivating empirical example of how land insecurity moderates the
effectiveness of land-based appeals by politicians from
\cite{horowitz2020can}.  We next describe the notation and basic
assumptions of the pre-test and post-test designs. We then derive the
sharp nonparametric bounds and sensitivity analyses for the pre-test,
post-test, and randomized placement designs. Next, we apply the
proposed methods to the empirical example. Finally, we conclude by
suggesting directions for future research.

\section{Motivating Example}\label{section_example}

We illustrate the trade-off between the pre- and post-treatment measurement of a
moderator using a survey experiment conducted by \cite{horowitz2020can}. This
study investigates the effectiveness of land-based appeals to increase a
politician's electoral support in Kenya's Rift Valley. We focus on two randomly
assigned conditions: a control condition in which participants heard a generic
campaign speech with no direct reference to the land issue, and a treatment
condition that additionally referenced the land issue.\footnote{To preserve
  statistical power, we focus on a treatment group that combines two conditions:
  one stating the land issues is their top priority and one that additionally
  reference an ethnic
  grievance.} 
The outcome is the participant's reported likelihood of supporting the candidate,
which we dichotomize (likely to support the candidate vs. not) to illustrate our proposed methods.

One of the main hypotheses tested by \cite{horowitz2020can} is whether individuals
who are personally experiencing land insecurity are more responsive to land-based
appeals by politicians. While the evidence from the full sample is inconclusive,
among respondents belonging to the ``insider'' ethnic group (Kalenjins), there is a positive
and statistically significant interaction between the treatment condition
and a dummy variable for land insecurity, in line with the authors' expectations.

Both types of bias could be of concern in this
context, as land insecurity is measured by asking respondents to rate
the security of their land rights.  Asking this question
\textit{before} treatment may raise the salience of the land issue and
thus attenuate any treatment effect of hearing a land-based appeal by
a politician, a form of priming bias. Conversely, if the
experimenter asks respondents about land security \textit{after}
treatment, then their responses may be affected by the content of the
speech, 
yielding post-treatment bias.

To address these concerns, \cite{horowitz2020can} use (what we call)
the \textit{randomized placement design}, in which questions relating to
the respondent's land rights were randomly assigned to occur before or
after the treatment.
Figure~\ref{fig:hk-plot1} shows the ``na\"{i}ve'' estimates of the
treatment-moderator interaction, which are all positive and substantively large, implying
that the effect of a land-based appeal on electoral support is 25--50 
percentage points larger for land-insecure respondents 
compared land-secure respondents. However, these estimates only reach
conventional levels of statistical significance with the increased
statistical power of the combined pre/post sample. 

\begin{figure}[t!]
    \centering
        \includegraphics[width=1\textwidth]{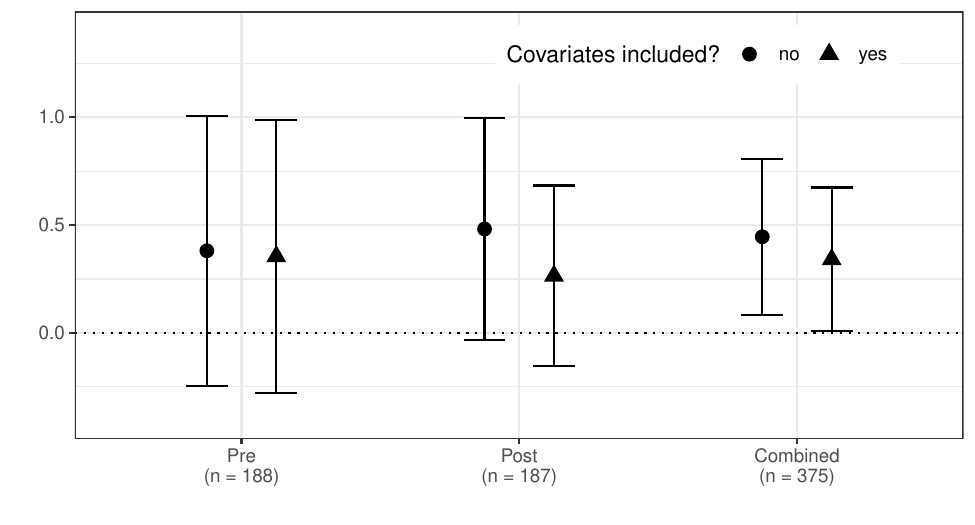}
        \caption{Estimates of treatment-moderator interaction using pre-test, post-test, and combined data from \cite{horowitz2020can}, with and without covariates (age, gender, education, and closeness to own ethnic group). Error bars indicate 95\% confidence intervals.}
        \label{fig:hk-plot1}
\end{figure}

What can we learn about the true interaction from these estimates? We show in
the following sections that the observed data often tells us very little about
interactions without strong assumption. In
Section~\ref{section_empirical_example}, we return to this empirical example to
illustrate how researchers can assess the possible impact of priming bias and
post-treatment bias on their substantive findings using our proposed methods.

\section{Experimental Designs and Causal Quantities of Interest}
\label{sec:setup}

We now lay out the formal notation of our setting and describe the
causal quantities of interest. Let $T_i$ represent the
binary treatment variable for unit $i \in \{1, \ldots, n\}$, indicating a certain
experimental manipulation received by the unit.  We use $Y_i$ to
denote a binary outcome of interest, and $M_i$ to represent an observed
binary moderator of interest.\footnote{We focus on the binary setting to ease exposition, though many of the methods could be extended to more general cases.}  Our goal is to develop a
methodology that can be used to understand how the average treatment effect of $T_i$ on
$Y_i$ varies as a function of $M_i$.

In this paper, we consider three experimental designs for this goal: the
pre-test design, the post-test design, and the randomized placement design.
These three designs differ in the timing of measurement for the potential
moderators. In the pre-test design, the experimenter measures the moderator,
$M_i$, before treatment assignment, ensuring that these measurements are
unaffected by treatment. In the post-test design, the experimenter measures
moderators after treatment assignment. Let $Z_i$ be an indicator for measuring $M_i$ before ($Z_i = 0$) or after treatment ($Z_i = 1$).

We will analyze these experimental designs using the potential outcomes
framework \citep[e.g.,][]{holl:86}. Let $Y_i(t, z)$
represent the potential outcome with respect to the treatment status $t$ and the
timing of covariate measurement, $z$. In our empirical application, for example,
$Y_{i}(1,1)$ is whether respondent $i$ would support the hypothetical candidate
if they were given the speech with referencing the land issue (that is,
treatment, $t=1$) and we measured the land security moderator \emph{after}
treatment ($z=1$).
We make the consistency assumption for the
potential outcomes, such that $Y_i = Y_i(T_i, Z_i)$.

Our goal is to estimate treatment effects on the outcome free of
priming bias, which requires us to measure moderators \emph{after}
treatment.
To this end, we define the ``true'' potential outcomes of interest as
$Y^{*}_{i}(t) \equiv Y_i(t, 1)$ and let $Y_{i}(t) \equiv Y_{i}(t,0)$ to be a
possibly mismeasured proxy observed under the pre-test design. Priming bias can
occur when the pre-test and post-test measurements of the outcome differ,
$Y_i(t) \neq Y^{*}_i(t)$.

Similarly, the measures of the moderator, self-reported land security,
can be affected by the when it is measured and the land grievance
framing treatment. Let $M_i(t, z)$ be the potential value of the
moderator that would be observed for unit $i$ when the treatment is
set to $t$ and the variable is measured in design $z$. 
Under the pre-test design, the treatment cannot affect the moderator, so we
refer to $M_i(0,0) = M_i(1,0) \equiv M^{*}_i$ as the ``true'' moderator for unit
$i$. On the other hand, under the post-test design, we let
$M_{i}(t) \equiv M_{i}(t, 1)$ to be potentially mismeasured proxy for that true
moderator. Under those designs, the treatment can affect respondents' moderator,
so that $M_i(0) \neq M_i(1)$, which can lead to post-treatment bias.

We can explore heterogeneous treatment effects by estimating the following
quantities of interest,
\begin{eqnarray}
  \tau(m) & \equiv & \E(Y_i^{*}(1) - Y^{*}_i(0) \mid M^{*}_{i} = m), \\
  \delta & \equiv & \tau(1) - \tau(0), \label{eq:delta}
\end{eqnarray}
where $m \in \{0,1\}$. The first quantity characterizes
the post-test conditional average treatment effect (CATE) as a
function of the pre-test value of the moderator. This CATE is the
effect of land-based appeals on support for a politician (when
unprimed) conditional on a level of true land insecurity. The second
quantity of interest compares the CATE between two
subpopulations with different levels of the pre-test moderator, which
is often called an \emph{interaction}.
These two quantities of interest formalize the dilemma about the
choice of pre-test versus post-test experimental designs. 
Identifying $\tau(m)$ requires observing the true potential outcomes of interest
($Y^{*}_i(t)$) and true moderator ($M^{*}_{i}$) for the same unit, but this can
never occur because the former requires $Z_{i} =1$ and the latter requires
$Z_{i} = 0$.

While the pre-test design allows us to observe the true moderator, it may suffer
from priming bias because asking questions about the moderator might change the
causal effect of the treatment by cueing respondents. Thus, under the pre-test
design, neither $\tau(m)$ nor $\delta$ nor even the ATE,
$\E(Y^{*}_i(1) - Y^{*}_i(0))$, is identified. In contrast, under the post-test
design, the ATE can be estimated without bias, and yet this design may result in
post-test bias for $\tau(m)$ and $\delta$ when the treatment affects the
moderator.

\section{Nonparametric Analysis of the Three Experimental Designs}
\label{sec:method}

We now show what we can learn from each of the three possible
experimental designs. Generally speaking, none of the three designs
are informative about the sign of the interaction effect without further
assumptions. For each design, we derive sharp
bounds for the interaction and show how to narrow these bounds with
additional substantive assumptions. We also develop a sensitivity
analysis procedure that allows researchers to vary the strengths of
such assumptions and assess their implications.

\subsection{Pre-test Design}
\label{subection_pretest_analysis}

We first investigate the identifying power of the pre-test design,
where we observe $(Y_i, M_i, T_i)$ among the units for whom $Z_i = 0$.
We first formalize the randomization of treatment in this setting. 
\begin{assumption}[Pre-test Randomization]
\label{a:randomization-pre}
\begin{eqnarray}
  \{Y_i(t, z), M_i(t, 1)\} \ \indep \ T_i & \mid & M^{*}_i = m, Z_i = 0,
\end{eqnarray}
for $t=0,1$, $m \in \{0,1\}$.
\end{assumption}
This assumption allows for the possibility of conditioning our
randomization on the pre-treatment moderator, though this assumption
also holds when randomization is unconditional. 

By computing the difference in the average observed outcomes between different treatment
groups, researchers can identify the following quantity,
\begin{eqnarray}
  \tau_{pre}(m) & \equiv &  \E(Y_i \mid T_i = 1, M_i = m, Z_i = 0) - \E(Y_i \mid
  T_i = 0, M_i = m, Z_i = 0) \nonumber \\
  & = & \E(Y_i(1) - Y_i(0) \mid M^{*}_i = m), \label{eq:tau-pre}
\end{eqnarray}
where the equality follows from Assumption~\ref{a:randomization-pre} and the
consistency assumption. Equation~\eqref{eq:tau-pre} does not generally equal the
true CATE, $\tau(m)$, since the conditional distribution of $Y_i(t)$ given
$M^{*}_i = m$ may differ from that of $Y^{*}_i(t)$ given $M^{*}_i = m$, since
asking questions about the moderator before measuring outcome may distort subsequent
responses. Thus, $\tau_{pre}(m)$ is the CATE of the land grievance treatment on the \emph{primed} support for the candidate. Similarly, $\delta$ is not generally equal to
$\delta_{pre} \equiv \tau_{pre}(1)-\tau_{pre}(0)$. The key problem here is that
the pre-test design gives us information about the correct values of the
moderator $M^{*}_{i}$, but provides no information about the outcomes we want to
investigate, $Y^{*}_{i}(t)$. 

With no restrictions on the priming bias, we cannot rule out extreme
possibilities, such as all respondents changing their values of $Y_{i}$ in
response to the priming. This could occur if asking about land security caused
all supporters to oppose the hypothetical candidate and vice versa. While
perhaps unlikely, this scenario is possible under the randomization
assumption alone.
Thus, the nonparametric bounds under
the pre-test design remain identical to the logical bounds,
\begin{eqnarray}
\tau(m) & \in & [-1,\ 1], \\
\delta & \in & [-2,\ 2].
\end{eqnarray}
In other words, the pre-test design is completely uninformative about
the causal effects of interest unless one is willing to make
additional assumptions about the joint distribution of $Y^{*}_i(t)$
and $Y_i(t)$.  

We can derive an expression for the amount of priming bias as
\begin{equation}
  \label{eq:priming-bias}
  \begin{aligned}
  \tau_{pre}(m) - \tau(m) &= \left\{\P(Y^{*}_{i}(1) = 0, Y_{i}(1) = 1 \mid M^{*}_{i} = m) - \P(Y^{*}_{i}(0) = 0, Y_{i}(0) = 1 \mid M^{*}_{i} = m)\right\} \\ & \quad- \left\{\P(Y^{*}_{i}(1) = 1, Y_{i}(1) = 0 \mid M^{*}_{i} = m) - \P(Y^{*}_{i}(0) = 1, Y_{i}(0) = 0 \mid M^{*}_{i} = m)\right\},
  \end{aligned}
\end{equation}
where we have suppressed the conditioning on $Z_{i} = 0$ to reduce
notational burden. The first term of this bias for the CATE is the effect of treatment on probability of being positively primed (moving from $Y^{*}_{i} = 0$ to $Y_{i}=1$) and the second term is the effect of treatment on being negatively primed (moving from $Y^{*}_{i}= 1$ to $Y_{i}=0$). Positively primed individuals in our empirical application would be those that would become supportive of the candidate only when asked about land security prior to treatment. We can see that this bias is unidentifiable because it depends on the joint distribution of the unprimed $Y_{i}^{*}(t)$ and primed $Y_{i}(t)$ outcomes, but we never observe both of these variables for the same unit.

Some scholars, such as \cite{SheCli23}, have shown empirically that
$\delta$ and $\delta_{pre}$ appear close to each other in several
experiments where the moderator was measured in a prior survey wave
(and so was less prone to priming bias). Should we conclude from this
that future studies will not suffer from priming bias? Perhaps, but we
should properly view this as an additional strong assumption rather than a result implied by the design of the experiment. In particular, we would have to assume that either there is no priming at all, treatment does not affect priming,  or that the treatment effects on positively and negatively primed individuals are exactly equal to each other and cancel out. Below, we will see how to incorporate assumptions about limited priming that would allow a sensitivity analysis for pretest designs.

\subsubsection{Narrowing the Pre-test Bounds under Additional Assumptions}

What assumptions can we place on the pre-test design to narrow the bounds? Any such assumption would have to place restrictions on the joint distribution of the primed and unprimed outcomes. A common set of assumptions for this type of setting would be a \emph{monotone treatment response} assumption for the priming effect \citep{Manski97}. In particular, we consider a priming
monotonicity assumption, which states that the effect of asking the
moderator before treatment can only move the outcome in a single
direction.
\begin{assumption}[Priming Monotonicity]\label{asp:prime}
  $Y_i(t) \geq Y^{*}_i(t)$ for all $t = 0,1$.
\end{assumption}
The assumption implies that the effect of moving from post-test ($Z_i = 1$) to
pre-test ($Z_i = 0$), which we call the priming effect, can only increase the
outcome. In our application, this assumption implies that asking
about land rights before reading the campaign speech treatment increases support
for the hypothetical candidate. As stated, this assumption holds across levels
of treatment, though it is possible to assume the reverse direction
($Y_i(t) \leq Y^{*}_i(t)$) or even to have a different effect direction for each
level of treatment. 

\begin{proposition}[Pre-test Sharp Bounds]
  Let $P_{tzm} = \P(Y_i = 1 \mid T_i = t, Z_i = z, M_i = m)$. Under
  Assumptions~\ref{a:randomization-pre} and~\ref{asp:prime}, we have
  the following sharp bounds:
  \begin{equation}
    \label{eq:1}
\tau(m) \in \left\{-P_{00m},\; P_{10m} \right\},
\end{equation}
and
\begin{equation}
  \label{eq:2}
  \delta \in \left\{-P_{100} - P_{001},\; P_{101} + P_{000}\right\}.
\end{equation}
\end{proposition}

Priming monotoncity narrows the bounds compared to the uninformative randomization bounds, but the bounds still will always contain zero. Thus, without further assumptions, the pre-test design cannot rule out no CATE for any level of the moderator nor can it rule out no interaction. 

\subsubsection{Sensitivity Analysis for the Pre-test Design}

Given the empirical results such as \cite{SheCli23} pointing to a limited role
of priming in some experimental settings, we now develop a sensitivity analysis
for pre-test designs to determine how sensitive results are to deviations from
the ``no priming bias'' assumption. In particular, we constrain the proportion
of respondents primed, or more precisely, the proportion of respondents whose
value of $Y^{*}_i(t)$ is different from $Y_i(t)$. Formally, we state this
restriction as
\begin{equation}
  \label{eq:theta_restriction}
  \Pr(Y^{*}_i(t)=1,Y_i(t)=0\mid M^{*}_i=m) + \Pr(Y^{*}_i(t)=0,Y_i(t)=1\mid M^{*}_i = m) \ \leq \ \theta
\end{equation}
for all $t, m \in\{0,1\}$. Note that if $\theta=0$, then we have
$Y^{*}_i(t)=Y_i(t)$ for all $i$, and the data under the pre-test design
alone identify the quantities of interest. By increasing the value of
$\theta$, a sensitivity analysis gradually restricts the severity of
the priming effects in the empirical setting.

\begin{proposition}[Pre-test Sharp Bounds under Restricted Priming Effects]\label{prop:pre_sens}
  Suppose that Assumptions~\ref{a:randomization-pre}
  and~\ref{asp:prime}, and
  restriction~\eqref{eq:theta_restriction} hold. Then, we have the
  following sharp bounds:
  \[
\tau(m) \in \left\{\tau_{pre}(m) - \theta, \tau_{pre}(m) + \theta\right\}, \qquad
\delta \in \left\{\delta_{pre} - 2\theta, \delta_{pre} + 2\theta\right\}.
\]
\end{proposition}
This proposition says that if the proportion of primed respondents in each
treatment arm is no larger than $\theta$, then we can bound the true interaction
with a $4\theta$ wide interval around the naive pre-test interaction (for the
CATE, the interval width is exactly half, i.e., $2\theta$). These bounds can be
informative (that is, they exclude 0) if $|\delta_{pre}| > 2\theta$.

\subsection{Post-test Design}
\label{subection_pt_analysis}

Next, we consider the post-test design, where we observe $(Y_i, M_i, T_i)$ among the units for whom $Z_i = 1$.  The randomization of the treatment implies the following ignorability assumption.
\begin{assumption}[Post-Test Randomization]
\label{a:randomization}
$$\{Y_i(t, z), M_i(t), M^{*}_i\} \ \indep \ T_i \mid Z_i = 1,$$
for $t = 0,1$.
\end{assumption}
How informative is Assumption~\ref{a:randomization} alone about the CATEs and
interaction without making any other assumption? The standard CATE estimator
would be unbiased for
\begin{eqnarray}
  \tau_{post}(m) & \equiv & P_{11m} - P_{01m} = \E[Y^{*}_i(1) \mid M_i(1)=m] - \E[Y^{*}_i(0) \mid M_i(0) = m],
\label{eq:tau-post}
\end{eqnarray}
where the equality follows from Assumption~\ref{a:randomization} and
the fact that $\Pr(Z_i=1)=1$ for all $i$.

Under the post-test design, we observe the true potential outcome of
interest $Y_i^\ast(t)$, and yet the moderator may be observed with
error, i.e., $M_i(t) \ne M_i^\ast$.  Similar to the case of the
pre-test design, therefore, $\tau_{post}(m)$ does not generally equal
$\tau(m)$ because the conditional distribution of $Y^{*}_i(t)$ given
$M_i(t) = m$ may differ from that of $Y^{*}_i(t)$ given $M^{*}_i = m$.
In fact, $\tau_{post}(m)$ may not equal $\tau(m)$
even if the effects of treatment and moderator measurement timing do
not change the marginal distribution of the moderator, i.e.,
$\Pr(M_i(0) = m)= \Pr(M_i(1) = m) = \P(M_{i}^{*} = m)$ for all
$m \in \mathcal{M}$. Restricting the {\it marginal} distribution of
the moderator does not help because effect heterogeneity is a function
of the {\it joint} distribution of the counterfactual outcomes and
moderators. Thus, under the post-test design, neither $\tau(m)$ nor
$\delta$ is nonparametrically identified.

Despite the lack of point identification, the design can contain some
information about the causal quantities of interest. The following proposition
derives the sharp (i.e., shortest possible) randomization-only bounds under the
post-test design.
\begin{proposition}[No-assumption Post-test Sharp Bounds]\label{prop:randomization-bounds}
  Let $P_{t} = \P(Y_{i} = 1 \mid T_{i} = t, Z_{i} = 1)$. Under Assumptions~\ref{a:randomization}, we have $\delta \in [\delta_{L}, \delta_{U}]$, where
\begin{align}
\delta_{L} & =  -1 - \max\left\{ \min\left(\frac{1-P_{0}}{P_{1}}, \frac{1-P_{1}}{P_{0}}\right), \min\left(\frac{P_{1}}{1 - P_{0}}, \frac{P_{0}}{1 - P_{1}}\right)\right\}, \label{eq:bounds-post-delta}\\
  \delta_{U}& =  1 + \min\left\{ \max\left(\frac{1 - P_{0}}{P_{1}}, \frac{1-P_{1}}{P_{0}}\right), \max\left(\frac{P_{1}}{1 - P_{0}}, \frac{P_{0}}{1-P_{1}}\right)\right\}.\nonumber
\end{align}
Furthermore, these bounds are sharp. 
\end{proposition}

The proof of this result is given in
Section~\ref{sec:randomization-bounds-proof} and rely on a standard
linear programming approaches often used in bounding causal quantities
\citep[e.g.,][]{balk:pear:97}. 
Unfortunately, these bounds are often quite wide in practice. In particular, the
bounds can never be narrower than $[-1, 1]$ since the post-test data are
completely uninformative about the true moderator. As a result, the sharp bounds
under the post-test design can be quite wide and sometimes cover the entire
possible range, $[-2,2]$, for $\delta$. In sum, without additional assumptions,
the post-test design can only provide limited information about the causal
quantities of interest.

\subsubsection{Narrowing the Post-test Bounds under Additional Assumptions}

While the bounds only using the randomization can be too wide to be useful in
practice, we may be willing to entertain other assumptions on the
causal structure that will narrow the bounds.
We rely on a principal stratification approach in which we stratify the units
according to how their moderator values react to treatment and moderator
measurement timing \citep{fran:rubi:02}. Let
$\mu_{s}(t) = \P[Y^{*}_i(t) = 1 \mid S_i = s]$, where $S_i$ represents the
principal strata defined by the moderator, $\{M_{i}(1), M_{i}(0), M^{*}_i\}$.
Without making any assumption, $S_i$ can take any of the $2^3$ values in
\[
\mathcal{S} = \{111, 011, 101, 001, 110, 010, 100, 000\}.
\]
Let $\rho_{s} = \Pr[S_i = s]$ be the probability of a unit falling
into one of the strata, such that $\sum_{s\in\mathcal{S}}\rho_s =
1$. 
Finally, we denote the marginal probability of the true pre-test moderator as $Q_* = \P(M^{*}_i = 1)$.

The first assumption we consider is that the effect of post-treatment
measurement of the moderator has a \emph{monotonic} effect on the
moderator for every unit.
\begin{assumption}[Moderator Monotonicity]\label{a:monotonicity}
$M_i(t) \geq M^{*}_i$ for all $t = 0,1$.
\end{assumption}
In the context of the motivating example, this assumption requires no unit to have (say) lower land security values if we ask about it after the respondent reads the politician's speech rather than before. While weaker than assuming there is no measurement error in the post-test moderator, the plausibility of this assumption will depend on
the experimental context. In the post-test design, we cannot verify
this assumption because we never observe $M_{i}^{*}$.  
The assumption rules out several possible principal strata, ensuring that $S_i$
can only take one of the following values: $111, 110, 010, 100,$ or $000$.\footnote{While
we present a positive version of monotonicity for both
treatment levels, it is possible to derive bounds under a negative version of the assumption or with differing directions for each treatment condition.} We now
present the sharp bounds under this assumption.

\begin{proposition}[Post-test Sharp Bounds under Monotonicity]\label{prop:monotonicity}
  Let $Q_{tz} = \Pr(M_i = 1 \mid T_i = t, Z_i = z)$. Under Assumptions~\ref{a:randomization}
and~\ref{a:monotonicity}, we have
sharp bounds $\delta \in [\delta_{L2}, \delta_{U2}]$, where
\[
  \begin{aligned} 
    \delta_{L2} &= \frac{P_{111} Q_{11} - P_{011}Q_{01}}{Q_*}  -
    \frac{P_{110} (1 - Q_{11}) -P_{010} (1 - Q_{01})}{1-Q_{*}} \\
    & \qquad + \frac{\max\{P_{011}Q_{01} - Q_*, 0\} - \max\{P_{111}Q_{11}, Q_{11} - Q_*\}}{Q_*(1 - Q_*)}, \\
    \delta_{U2} &= \frac{P_{111} Q_{11} - P_{011}Q_{01}}{Q_{*}} -
    \frac{P_{110} (1 - Q_{11}) -P_{010} (1 - Q_{01})}{1-Q_{*}}\\
    &\qquad + \frac{\min\left\{0, Q_* - P_{111}Q_{11}\right\} + \min\left\{ P_{011}Q_{01}, Q_{01} - Q_* \right\}}{Q_*(1 - Q_*)}.
  \end{aligned}
\]
\end{proposition}
We provide the derivation of these bounds (and those in the next
proposition) in Section~\ref{sec:bounds-proof}. 
These bounds will differ from the above randomization bounds given in
Proposition~\ref{prop:randomization-bounds} in two ways. First, with
the randomization assumption alone, we could only leverage the
observed strata within levels of treatment---further stratification in
terms of the moderator provided no information because it placed no
restriction on the relationship between the pre-test and post-test
versions of the moderator.  Under moderator monotonicity
(Assumption~\ref{a:monotonicity}), we can leverage $P_{tzd}$ and
$Q_{tz}$ to narrow the bounds. Second, moderator monotonicity places bounds on
the true value of $Q_*$ since it must be less than
$\min(Q_{11}, Q_{01})$.

While the moderator monotonicity assumption does narrow the bounds,
they are often still quite wide and usually contain 0.  To further
narrow the bounds, we consider another assumption that the moderator
is stable in the control arm of the study.
\begin{assumption}[Stable Moderator under Control]\label{a:stability}
$M^{*}_i = M_i(0)$
\end{assumption}
\noindent This assumption implies that the moderator under control in the
post-test design is the same as the moderator as if it was measured pre-test.
This assumption may be plausible in experimental designs where the control
condition is neutral or similar to the pre-test environment. In our empirical
example, this would mean that hearing the generic campaign speech, which does
not mention the land issue, does not affect perceived land insecurity. Under
both Assumptions~\ref{a:monotonicity}~and~\ref{a:stability}, the only values
that principal strata that $S_i$ can take are $\{111, 100, 000\}$.

\begin{proposition}[Post-test Sharp Bounds under Moderator Monotonicity and Stability]\label{prop:stability}
  Under Assumptions~\ref{a:randomization}, \ref{a:monotonicity}, and~\ref{a:stability}, we have $Q_* = Q_{01}$ and sharp bounds $\delta \in [\delta_{L3}, \delta_{U3}]$, where
  \[
  \begin{aligned}
    \delta_{L3} &= \frac{P_{111}Q_{11}}{Q_{01}} - P_{011} - \frac{P_{110}(1 - Q_{11})}{1-Q_{01}}  + P_{010} - \min\left\{ 1, \frac{P_{111}Q_{11}}{Q_{11} - Q_{01}}  \right\} \cdot \frac{Q_{11} - Q_{01}}{Q_{01}(1-Q_{01})}, \\
    \delta_{U3} & = \frac{P_{111}Q_{11}}{Q_{01}} - P_{011} -
    \frac{P_{110}(1 - Q_{11})}{1-Q_{01}}  + P_{010} - \max\left\{ 0,
      \frac{ P_{111}Q_{11} - Q_{01}}{Q_{11} - Q_{01}} \right\} \cdot \frac{Q_{11} - Q_{01}}{Q_{01}(1-Q_{01})}.
  \end{aligned}
\]
\end{proposition}
These bounds demonstrate how the magnitude of treatment effect on the
moderator affects identification in the post-test design. A unit's
moderator is affected by treatment whenever $M_i(1) \neq M_i(0)$,
which corresponds to the $S_i = 100$ principal stratum under
monotonicity and stable moderator under control.  Note that under
these assumptions, $\rho_{100}$ represents the magnitude of the
treatment-moderator effect. Since $Q_{11} = \rho_{111} + \rho_{100}$
and $Q_{01} = \rho_{111}$, we can identify this effect with the usual
difference in (population) means, $Q_{11} - Q_{01} = \rho_{100}$. The
maximum possible width of the sharp bounds depends on this effect,
with
\[
\max\{\delta_{U3} - \delta_{L3}\} = \frac{Q_{11} - Q_{01}}{Q_{01}(1 - Q_{01})} = \frac{\rho_{100}}{\rho_{000}(\rho_{111} + \rho_{100})},
\]
so that the bounds can be relatively informative if the post-treatment
average effect on the moderator is small.

\subsubsection{Sensitivity Analysis under Limited Effects on the Moderator}\label{subsection-post-sens}

While the monotonicity and stable moderator assumptions can considerably narrow
the nonparametric bounds on our causal quantities, they rule out entire
principal strata, which may be stronger than is justified for a particular
empirical setting. We now consider an alternative approach to bounds that does
not rule out any particular principal strata but rather places restrictions on
the proportion of units whose moderators are affected by treatment.

In particular, we propose a sensitivity analysis that limits the
proportion of respondents whose moderator value changes between the
pre-test and post-test, regardless of the treatment condition
(contrast this with Assumption~\ref{a:stability} which applies to the
control condition only).  We operationalize this via the following
constraint,
\[
\Pr(S_i \notin \{111, 000\}) \leq \gamma.
\]
Note that $\gamma$ must be greater than $|Q_{11} - Q_{01}|$ for the bounds to be
feasible since
$|Q_{11} - Q_{01}| = |\rho_{101} + \rho_{100} - \rho_{011} - \rho_{010}|$. We
vary the value of $\gamma$ from $|Q_{11} - Q_{01}|$ to $1$ and see how the
nonparametric bounds on the value of $\delta$ change as we gradually allow a
larger treatment effect on the moderator. We obtain these new bounds by adding
this additional constraint to linear program, which we can then solve
numerically. This approach is a more flexible way to allow for limited
heterogeneous treatment effects on the moderator in any direction. Researchers
can also combine this sensitivity analysis with the monotonicity and stable
moderator assumptions.

\subsection{Randomized Placement Design}

Finally, we examine a combined pre/post design called the randomized
placement design, where in addition to treatment, the timing of
moderator measurement, $Z_i$, is also randomized. Bounds from this design will improve upon the post-test bounds because we can identify the marginal probability of the true moderator
as
$$ Q_*  \ = \ \Pr(M^{*}_i = 1) \ = \ \Pr(M_i=1\mid Z_i=0).$$
Unfortunately, for the randomization-only bounds in
Equation~\eqref{eq:bounds-post-delta}, the sign of $\delta$ cannot be
identified for any value of $Q_* \in (0,1)$.

With the randomized placement design, we have a slightly more
complicated set of principal strata since now we must handle both
the pre-test and post-test potential outcomes. In particular, we have
the following:
\[
  \psi_{y_1y_0s}(t) \ = \ \Pr(Y^{*}_i(t)=y_1, Y_i(t)=y_0 \mid S_i = s)
\]
where $\psi_{y_1y_0s}(t) \geq 0$ and $\sum_{y_1}\sum_{y_0}
\psi_{y_1y_0s}(t) = 1$ for all $t$. These values characterize the joint
distribution of the pre-test and post-test potential outcomes for a
given treatment level, $t$, and principal strata, $s$. Furthermore,
let $\mathcal{S}^*_{m} = \{s : s \in \mathcal{S}, M^{*}_{i} = m\}$ be the set of principal strata
with the true value of the moderator equal to $m$.

Given these, we can write the interaction between the treatment and the moderator as
$$ \delta \ = \ \sum_{y_0=0}^1  \left\{ \sum_{s_1\in \mathcal{S}^*_{1}} \frac{\rho_{s_1}}{Q_*}(\psi_{1y_0s_1}(1) - \psi_{1y_0s_1}(0)) - \sum_{s_0\in \mathcal{S}^*_{0}}\frac{\rho_{s_0}}{1- Q_*}(\psi_{1y_0s_0}(1) - \psi_{1y_0s_0}(0)) \right\},$$
 and the observed strata in the pre-test and post-test arms as
\[
  \begin{aligned}
    P_{t1m}Q_{t1} &= \Pr(Y_i= 1, D_i=m \mid T_i=t, Z_i=1) \ = \ \sum_{y_0=0}^1 \sum_{s\in\mathcal{S}(t, 1, m)}\psi_{1y_0s}(t)\rho_{s}  , \\
    P_{t0m_{\ast}} & = \Pr(Y_i=1 \mid T_i = t, Z_i = 0, M_i = m_\ast) = \frac{\sum_{y_1=0}^1 \sum_{s\in \mathcal{S}^*_{m_{\ast}}} \psi_{y_11s}(t)\rho_{s}}{\sum_{s\in \mathcal{S}^*_{m_{\ast}}} \rho_{s}},
  \end{aligned}
\]
for all values of $y$, $m$, $t$, and $m_{\ast}$.  

Without further assumptions, the pre-test data are helpful only
insofar as they identify the marginal distribution of the moderator,
$Q_*$. 
We can narrow the bounds under the random placement design using all of the substantive assumptions described above: priming monotoncity, moderator monotonicity, and stability under control. Each of these implies restrictions on the above principal strata that can be incorporated into a linear programming problem that can be solved numerically to provide the bounds.

\subsubsection{Sensitivity Analysis in the Random Placement Design}

The random placement design allows us to combine the sensitivity analysis
procedure of the pre- and post-test designs. In particular, we can impose both
of the restrictions from above simultaneously:
\[
  \begin{aligned}
     \Pr(S_i \notin \{111, 000\}) &\leq \ \gamma, \\
    \Pr(Y^{*}_i(t)=1,Y_i(t)=0\mid M^{*}_i=m_\ast) + \Pr(Y^{*}_i(t)=0,Y_i(t)=1\mid M^{*}_i = m_\ast) \ &\leq \ \theta
  \end{aligned}
\]
Again, these restrict (a) how much the treatment and moderator measurement
timing affect the moderator, and (b) how much the moderator measurement timing
affects the outcome. To incoprate these restrictions into the bounds, we can
rewrite them in the principal strata described above:
\[
  \begin{aligned}
    (1 - \rho_{111} - \rho_{000}) &\leq \ \gamma, \\
    \frac{\sum_{s\in \mathcal{S}^*_{m_{\ast}}} \left( \psi_{10s}(t) + \psi_{01s}(t) \right) \rho_{s}}{\sum_{s\in \mathcal{S}^*_{m_{\ast}}} \rho_{s}} \ &\leq \ \theta.
  \end{aligned}
\]
We can easily add these restrictions to the optimization problem that produces
the bounds on the interaction.

We can conduct sensitivity analyses on the randomized placement design
by varying the values of $\gamma$ and $\theta$ and seeing how the
values of the bounds change. There are several ways to conduct and
present such a two-dimensional sensitivity analysis. One would be to
plot the parameters on each axis and demarcate the regions where the
bounds are informative (e.g., do not include zero) and where they are not. A
second approach would be to choose a small value for one of the two
parameters consistent with a researcher's beliefs. For instance, if a
researcher believes that the moderator is unlikely to be affected by
treatment, then they could choose a small value for $\gamma$ and
investigate the sensitivity of the bounds to different amounts of
priming, as measured by $\theta$. 

\subsection{Statistical Inference}
\label{subsection_statistical_inference}

The above bounds are stated in terms of population quantities when, in fact, we
only ever have sample data. We can easily obtain estimates for the bounds by
plugging in sample versions of these probabilities or, for the random placement
design, solving the linear programming problem using the sample data. Obtaining
valid confidence intervals in this setting is more challenging, since standard
asymptotic analyses break down due to the maximum and minimum operators causing
nondifferentiability. This problem can lead the standard nonparametric bootstrap
to have problematic theoretical properties \citep{AndHan09, FanSan18}.

Furthermore, confidence intervals for the bounds tend to be overly conservative
when the target of inference is the parameter $\delta$ rather than the bounds
themselves. As pointed out by \citet{ImbMan04}, this occurs because the true
parameter cannot simultaneously be close to the upper and lower bound at the
same time. This fact allows us to narrow the confidence intervals by a data-driven amount
while maintaining nominal coverage for the parameter of interest.

Our approach to inference combines the \cite{ImbMan04} approach with estimated
standard errors of the bounds from the the nonparametric bootstrap. While the
bootstrap may have theoretical problems, we find in simulations that this
approach produces conservative confidence intervals that have slightly
higher-than-nominal empirical coverage. 
In Supplemental Materials~\ref{sec:estimation_details}, we provide further details of our
approach to estimation. In Supplemental Materials~\ref{section_mcmc}, we also
develop a Bayesian model-based approach to incorporate additional pre-treatment
covariates that might be available to researchers.

\section{Empirical Example}\label{section_empirical_example}

To illustrate how our approach can be applied to both the post-test and
randomized placement design, we return to the example from
\cite{horowitz2020can} introduced in Section~\ref{section_example}. Recall that
the interaction in this case shows how the effect of land-based appeals varies
by a respondent's level of land security.

\subsection{Setup and Assumptions}

First, it is worth discussing the assumptions beyond randomization in this context.
Assumptions~\ref{a:monotonicity} (moderator monotonicity) and~\ref{a:stability}
(stability) are not guaranteed by the design and require a substantive
justification. In this case, moderator monotonicity requires the placement of
the land insecurity question after treatment shift perceived land insecurity in
the same direction for all respondents (or have no effect). We assume a positive
(or zero) individual effect on land insecurity, consistent with the estimate
ATE---though it is substantively small and not statistically significant
($\hat{\beta} = 0.04$, $p = 0.23$). Given that one part of the active treatment
was intended to induce fear of ``land grabbing,'' it seems somewhat plausible
that measuring perceived land rights after treatment would only increased
feelings of insecurity.

The assumption of stability means that hearing the generic campaign
speech (with no appeals to the land issue) has no
individual-level effect on perceived land insecurity when it is
measured post-treatment. Again, this assumption cannot be conclusively
tested with the observed data since we cannot estimate
individual-level effects. However, with the randomized placement
design, we can estimate the ATE of the pre/post randomization on the
moderator among respondents in the control group.  Our estimate of the
ATE is very small and not statistically significant
($\hat{\beta} = -0.005$, $p = 0.91$), which is at least consistent
with the assumption of a stable moderator under control.

\subsection{Comparing the Sharp Bounds across Different Assumptions}

\begin{figure}[t!]
    \centering
        \includegraphics[width=1\textwidth]{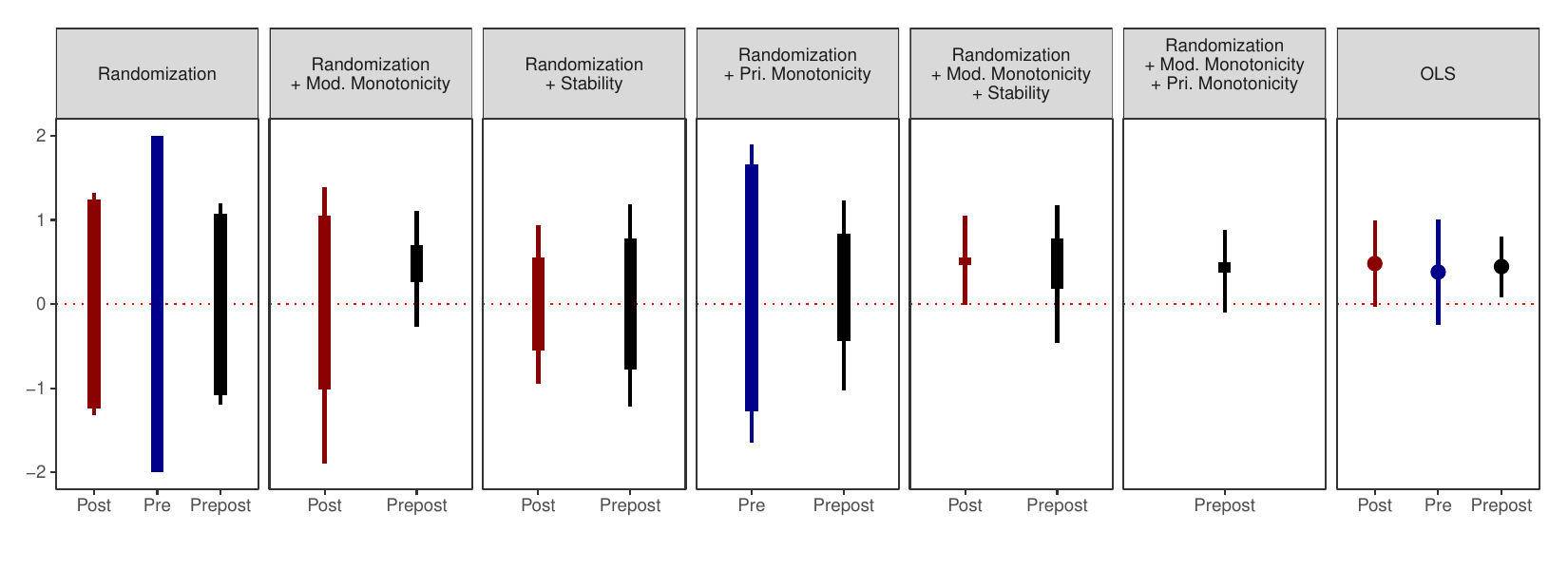}
        \caption{Estimated nonparametric bounds (thick bars) and 95\% confidence intervals (thin bars) under different designs and assumptions. The final panel contains OLS point estimates and 95\% confidence intervals. 
          }
        \label{fig:hk-plot2-non-param}
\end{figure}

Figure~\ref{fig:hk-plot2-non-param} displays the non-parametric bounds (with
95\% confidence intervals) for $\delta$ under different sets of assumptions
applied to the post-test data (in grey), the pre-test (in blue), and the
combined random placement design (labeled ``Prepost'' in black). We also include
the na\"{i}ve OLS estimate with 95\% confidence interval for comparison. Some of
the designs are missing from panels because an assumption does not apply to that
design (for example, moderator monotonicity under the pre-test design).

Assuming only randomization, the nonparametric bounds are uninformative of the
sign of $\delta$, and are much wider than the confidence interval of the
na\"{i}ve OLS estimate for all of the designs. Adding the assumption of
moderator monotonicity reduces the width of the bounds, especially on the
combined prepost data. Unfortunately, the confidence intervals become
considerably larger, especially for the post-test data, in which the $M_{i} = 1$
group is small.

The assumption of stability tightens the bounds to a similar degree for the
post-only and pre-post data. Under the pre-test assumptions of randomization,
moderator monotonicity, and stability, the bounds exclude zero for both the
post-test and random placement designs, though the confidence interval for the
latter contains zero. Furthermore, the bounds for this design are wider under
all three assumptions than under just randomization and monotonicity since, in
finite samples, the pre-test mean of $M_{i}$ differs from the post-test mean of
$M_{i}$ in the control arm. We would expect a difference by random chance even
when stability holds, but the bounds may widen slightly to accommodate this
divergence between the population and sample quantities. 

Priming monotonicity also narrows the bounds for the pre-test and random
placement designs, though not by as much as moderator monotonicity. The
combination of the two monotonicity assumptions, however, recovers bounds that
are close to the OLS estimate and has confidence intervals that barely contain
zero.

In sum, the nonparametric bounds do not support the hypothesis of a positive
interaction effect under the randomization assumptions alone. It is only with
a combination of additional substantive assumptions---moderator monotonicity,
priming monotonicity, and stability---that the sharp bounds produce results
qualitatively similar to the na\"{i}ve OLS estimate of a positive interaction
effect, though the designs differ on whether this is statistically significant
or not.

\subsection{Implementing the Sensitivity Analysis}

We now apply our sensitivity analysis procedures to this experiment. These procedures involve varying the proportion of respondents for whom the placement of the moderator measure affects their moderator value (land security) or their outcome (candidate support), labeled $\gamma$ and $\theta$ respectively. We can apply the $\gamma$ and $\theta$ sensitivity analyses to the post-test and pre-test designs, respectively, and we can combine them in the random placement design.

Figure~\ref{fig:hk-plot22} shows the post-test bounds as a function of $\gamma$
under just the randomization assumption. While $\gamma$ can theoretically range
up to 1, here we limit it to 0.5 to aid presentation since the bounds quickly
stabilize. The black lines denote the upper and lower bounds, and the shaded
ribbon denotes the 95\% confidence intervals around the bounds. The lower bound
crosses 0 when $\gamma = 0.07$---that is, when no more than 7\% of respondents
are affected by the post-treatment measurement of the moderator. The 95\%
confidence interval contains 0 even for the minimum possible value of $\gamma$
consistent with the observed data (0.02). The sensitivity analysis shows that
the sharp bounds are highly sensitive to changes in the degree of post-treatment
bias for small values of $\gamma$.

\begin{figure}[t!]
  \centering
  \includegraphics[width=0.8\textwidth]{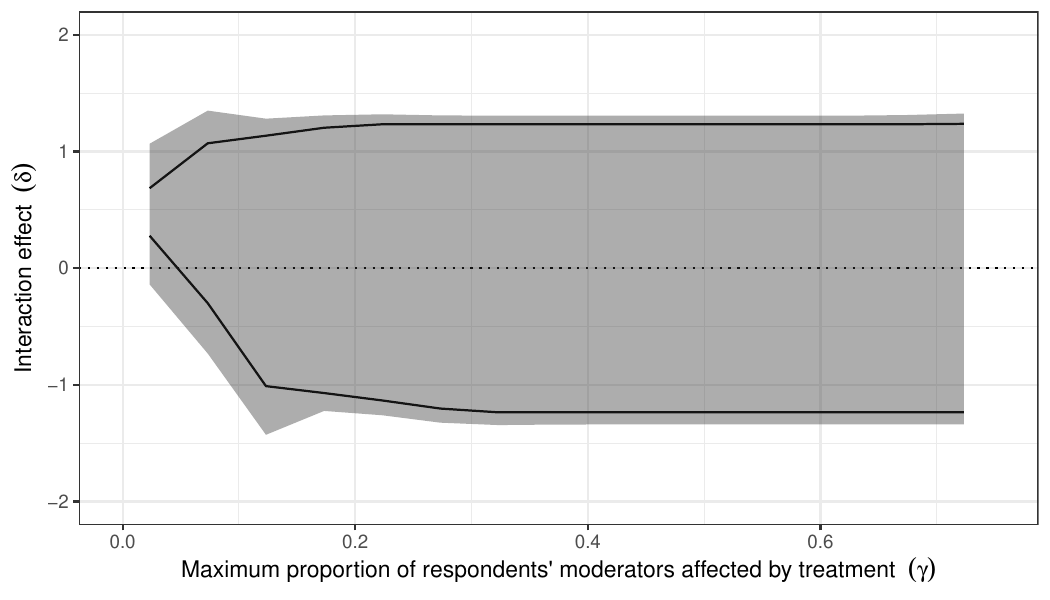}
  \caption{Post-test sensitivity analysis. Nonparametric bounds (black lines) with 95\% confidence intervals (grey ribbon) as a function of $\gamma$, the proportion of respondents whose value of the moderator variable (land insecurity) is affected by post-test measurement.}\label{fig:hk-plot22}
\end{figure}

To interpret this sensitive analysis, researchers will need to draw on their
substantive knowledge to assess the plausible range of $\gamma$. The estimated
ATE on the post-treatment moderator is $0.02$ ($p = 0.59$), but without a
monotonicity assumption this may include respondents with positive and negative
effects that offset their effects. We might assume that most of the effect of
the land rights prime would accrue to respondents with less certain responses
such as feeling ``somewhat secure'' in their land rights and that had been
personally affected by prior ethnic conflict. Using the pre-test data, we find
that this is 12\% of the sample, which we might consider a reasonable value of
$\gamma$. While it may seem like a relatively minor problem if only 12\% of the
sample is affected, our sensitivity analysis shows that the nonparametric bounds
would be about eight times wider ($[-1.02, 1.14]$) compared to the case where
$\gamma$ is at its minimum ($[0.38, 0.64]$).

\begin{figure}[t!]
  \centering
  \includegraphics[width=0.8\textwidth]{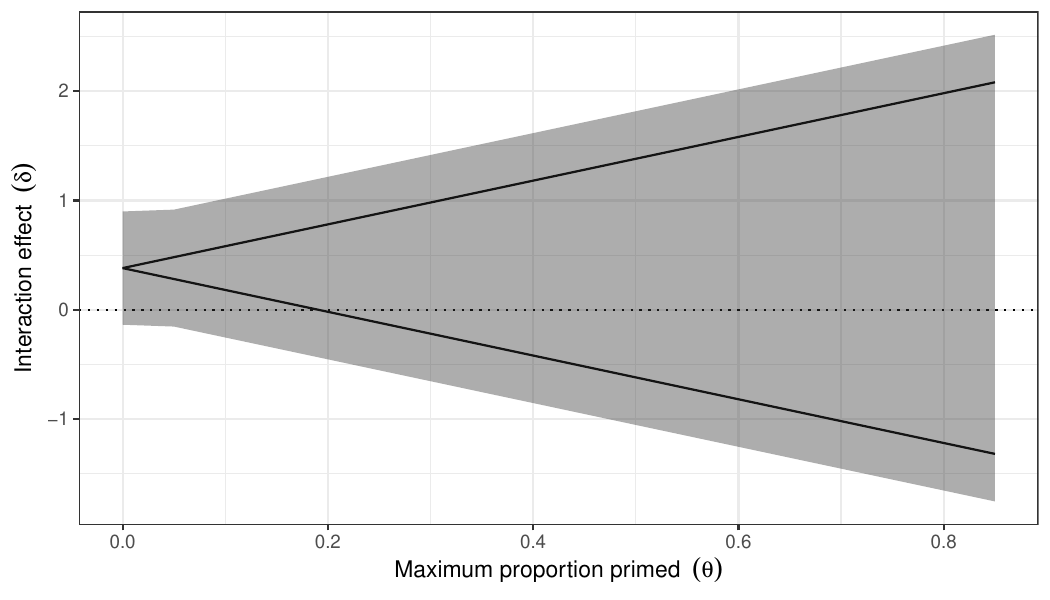}
  \caption{Pre-test sensitivity analysis. Nonparametric bounds (black lines) with 95\% confidence intervals
    (grey ribbon) as a function of $\theta$, the proportion of respondents
    who are primed by asking the moderator before treatment, under priming monotoncity assumption.} \label{fig:hk-plot22-pre}
\end{figure}

For priming bias in the pre-test design, we can assess sensitivity as a function
of $\theta$, the proportion of respondents whose outcome value is affected by
when the moderator is measured. Figure~\ref{fig:hk-plot22-pre} shows the bounds
as a function of $\theta$ under priming monotonicity. When $\theta = 0$, we are
assuming that priming bias does not exist so the true interaction is point
identified in the pre-test arm, though the confidence interval at $\theta = 0$
already includes zero. The estimated lower bound crosses zero at approximately
0.2, when up to 20\% of respondents are primed, which would be a rather large
effect of priming.

Finally, in the random placement design, we combine these two tests by restricting both $\gamma$ and $\theta$. Figure~\ref{fig:hk-prepost-sens} shows the prepost bounds as a function of $\gamma$ under two different assumption about the amount of priming: limited priming $\theta \leq 0.25$ and unrestricted priming $\theta \leq 1$. In both settings, the bounds initially shrink as $\gamma$ increases, most likely to due to the low values of $\gamma$ being inconsistent with the data. The bounds then begin to widen, though they widen much more for the unrestricted priming compared the limited priming setting. 

\begin{figure}[t!]
  \centering
\includegraphics[width=0.8\textwidth]{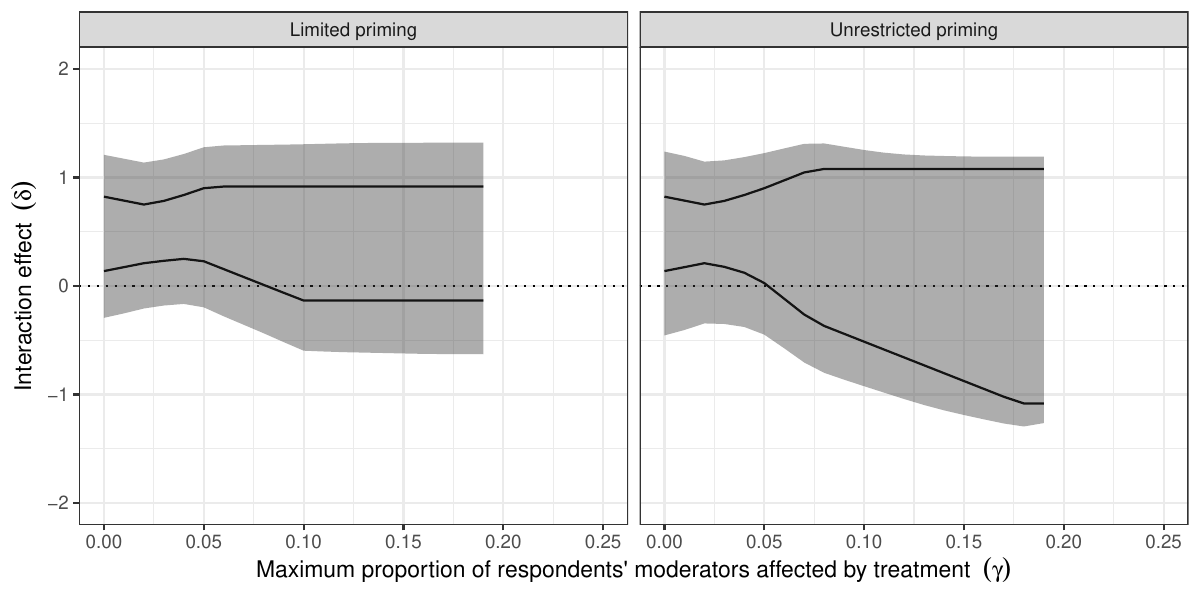}
  \caption{Random placement design sensitivity analysis. Nonparametric bounds (black lines) and 95\% confidence intervals (grey ribbons) as a function of the post-test effect on the moderators and the amount of priming. The limited priming assumption assumes $\theta \leq 0.25$ and the unrestricted priming has $\theta \leq 1$.}
  \label{fig:hk-prepost-sens}
\end{figure}

\section{Concluding Remarks}

This paper addresses a central tension in survey methodology: how should researchers assess priming bias versus post-treatment bias when designing a survey experiment? We provide sharp bounds for interactions with moderators measured post-treatment and show how these bounds vary under additional substantive assumptions. We also provide sensitivity analyses for both types of bias by varying the proportion of respondents whose moderator value changes in the post-test design and the proportion of respondents for whom the pre-test measurement of the moderator would prime their responses. We demonstrate how these tools can be used to diagnose and assess the severity of post-treatment bias and priming bias by applying them to a survey experiment regarding the effect of land-based appeals by politicians on electoral support in Kenya.

Open questions remain from our approach here. In particular, future work could optimize the random placement design to balance the priming and post-treatment bias concerns. In addition, we could consider how integrating separate pre-test surveys, often given weeks or months before treatment, might allow for a different set of plausible assumptions and identification.

\pdfbookmark[1]{References}{References}
\bibliographystyle{apsr}
\bibliography{prepost}

\if0\submission{}
\clearpage
\begin{appendices}
\renewcommand{\thesection}{\Alph{section}}
\renewcommand\thefigure{SM.\arabic{figure}}
\renewcommand\thetable{SM.\arabic{table}}
\setcounter{lemma}{0}
\renewcommand{\thelemma}{SM.\arabic{lemma}}
 
\section{Proofs}

\subsection{Pre-test bounds}
\label{sec:pre-test-proofs}

\begin{proof}[Proof of Proposition~\ref{prop:monotonicity}]
  Note that we can write $\tau_{pre}(m) = P_{10m} - P_{00m}$. From expression~\ref{eq:priming-bias}, we can write the true CATE under priming monotonicity as
  \[
  \begin{aligned}
   \tau(m) &=  P_{10m} - P_{00m} - \left\{\P(Y^{*}_{i}(1) = 0, Y_{i}(1) = 1 \mid M^{*}_{i} = m) - \P(Y^{*}_{i}(0) = 0, Y_{i}(0) = 1 \mid M^{*}_{i} = m)\right\}.
  \end{aligned}
  \]
  We can bound the unknown probabilities as
  \[
  0 \leq \P(Y^{*}_{i}(t) = 0, Y_{i}(t) = 1 \mid M^{*}_{i} = m) \leq P_{t0m}.
  \]
  Thus, a sharp upper bound on $\tau(m)$ would because
  \[
    P_{10m} - P_{00m} + P_{00m} = P_{10m},
  \]
  and a sharp lower bound would be
  \[
    P_{10m} - P_{00m} - P_{10m} = -P_{00m},
  \]
  which establishes the bounds for $\tau(m)$. For the upper $\delta = \tau(1) - \tau(0)$, we simply use the upper bound for $\tau(1)$ and the lower bound for $\tau(0)$ and vice versa for the lower bound for $\delta$.
\end{proof}

\begin{proof}[Proof of Proposition~\ref{prop:pre_sens}]
  Under priming mononoticity, we can write the true CATE as
  \[
  \begin{aligned}
   \tau(m)  &= \tau_{pre}(m) -\left\{\P(Y^{*}_{i}(1) = 0, Y_{i}(1) = 1 \mid M^{*}_{i} = m) - \P(Y^{*}_{i}(0) = 0, Y_{i}(0) = 1 \mid M^{*}_{i} = m)\right\}.
  \end{aligned}  
\]
The restriction on the amount of priming bias implies that
\[
0 \leq \P(Y^{*}_{i}(t) = 0, Y_{i}(t) = 1 \mid M^{*}_{i} = m) \leq \theta,
\]
for all $t$. Thus, clearly we have
\[
\tau(m) \in \left[\tau_{pre}(m) - \theta, \tau_{pre}(m) + m\right]. 
\]
Taking the maximum and minimum of $\tau(1)$ and $\tau(0)$, respectively, establishes the upper bound for $\delta$. Reversing this gives the lower bound. 

\end{proof}

\subsection{Post-test randomization bounds}
\label{sec:randomization-bounds-proof}

\begin{proof}[Proof of Proposition~\ref{prop:randomization-bounds}]

  Under Assumption~\ref{a:randomization}, the information about the
  parameter of interest comes from
  $P_{t} = \P(Y_{i} = 1 \mid T_{i} = t, Z_{i} = 1)$ alone.  This is
  because the distribution of the post-test moderators provides no
  information about the pre-test moderator. Recall that
\begin{equation}
  \label{eq:pt-pi}
P_{t} = \pi_{t1}Q_{*} + \pi_{t0}(1 - Q_{*}),
\end{equation}
where $\pi_{tm} = \P[Y_{i}(t, 1) = 1 \mid M^{*}_i = m]$ and $Q_{*} = \P[M^{*}_i = 1]$.

Below, we show how to derive the upper bound for $\delta$. The
derivation of the lower bound is similar.  Conditional on $Q_{*}$, we
can define the following linear program:
\[
  \begin{aligned}
    \max \; & \pi_{11} - \pi_{01} - \pi_{10} + \pi_{00} \\
    \text{subject to} \;& \pi_{t1}Q_{*} + \pi_{t0}(1- Q_{*}) = P_{t}
    \quad \text{for } t = 0,1,\\
    & 0 \leq \pi_{tm} \leq 1 \quad \forall (t, d) \in \{0,1\}^{2}
  \end{aligned}
\]
We can convert this to an augmented form by adding slack variables,
\[
  \begin{aligned}
    \max \; & \pi_{11} - \pi_{01} - \pi_{10} + \pi_{00} \\
    \text{subject to} \;& \pi_{t1}Q_{*} + \pi_{t0}(1- Q_{*}) = P_{t}
    \quad \text{for } t = 0,1 \\
            &  \pi_{tm} + s_{tm} = 1 \quad  \forall (t, m) \in \{0,1\}^{2} \\
    & \{\pi_{11}, \pi_{01}, \pi_{10}, \pi_{00}, s_{11}, s_{01}, s_{10}, s_{00}\} \geq 0.
  \end{aligned}
\]

The feasibility of various basic solutions here will depend on the relationship between the observed probabilities and $Q_{*}$.  In Table~\ref{t:bfs}, we show basic feasible solutions for the four different conditions relating $P_{1}$ and $P_{0}$ to $Q_{*}$. Under each condition, it is straightforward to determine that the basic feasible solution is also optimal since there is no entering variable that can increase the value of the quantity of interest. Thus, we know that
\[
\delta \leq \min\left\{\frac{1+P_1-P_0}{Q_*}, \ \frac{1-P_1+P_0}{1-Q_*}, \ \frac{1-P_0}{Q_*} + \frac{1-P_1}{1-Q_*}, \ \frac{P_1}{Q_*} + \frac{P_0}{1-Q_*} \right\}
\]

\begin{table}
  \caption{Optimal solutions to the linear program under different conditions}\label{t:bfs}
\small
\begin{tabular}{l|rrrrrrrr|r}
  Condition & $\pi_{11}$& $\pi_{10}$& $\pi_{01}$& $\pi_{00}$& $s_{11}$& $s_{01}$& $s_{10}$& $s_{00}$ & Value \\ \hline
  $P_{1} > Q_{*}, P_{0} > 1 - Q_{*}$ & $\frac{P_{1}}{Q_{*}}$ & 0 & 0 & $\frac{P_{0}}{1 - Q_{*}}$ & $1 - \frac{P_{1}}{Q_{*}}$ & 1 & 1 & $1 - \frac{P_{0}}{1 - Q_{*}}$ & $\frac{P_{1}}{Q_{*}} + \frac{P_{0}}{1-Q_{*}}$ \\
  $P_{1} < Q_{*}, P_{0} > 1 - Q_{*}$ & $1$ & $\frac{P_{1} - Q_{*}}{1-Q_{*}}$ & 0 & $\frac{P_{0}}{1 - Q_{*}}$ & $0$ & 1 & 1 & $1 - \frac{P_{0}}{1 - Q_{*}}$ & $\frac{1 - P_{1} + P_{0}}{1 - Q_{*}}$ \\
  $P_{1} > Q_{*}, P_{0} < 1 - Q_{*}$ & $\frac{P_{1}}{Q_{*}}$ & 0 & $\frac{P_{0} -1 + Q_{*}}{Q_{*}}$ & $1$ & $1 - \frac{P_{1}}{Q_{*}}$ & 1 & 1 & $0$ & $\frac{P_{1} - P_{0} + 1}{Q_{*}}$ \\
  $P_{1} < Q_{*}, P_{0} < 1 - Q_{*}$ & $1$ & $\frac{P_{1} - Q_{*}}{1-Q_{*}}$ & $\frac{P_{0} -1 + Q_{*}}{Q_{*}}$ & $1$ & $0$ & 1 & 1 & $0$ & $\frac{1 - P_{1}}{1 - Q_{*}} + \frac{1 - P_{0}}{Q_{*}}$ \\  
\end{tabular}
\end{table}
\normalsize
\noindent A similar derivation shows that
\[
\delta \geq \max\left(-\frac{1-P_{1}+P_{0}}{Q_*}, \ -\frac{1+P_{1}-P_{0}}{1-Q_*}, \ -\frac{P_0}{Q_*} - \frac{P_1}{1-Q_*}, \ -\frac{1-P_1}{Q_*} - \frac{1-P_0}{1-Q_*}\right).
\]

Taking a look at these bounds, suppose that $P_{1} \geq 1-P_{0}$, then the upper bound as a function of $Q_{*}$ is:
\[
  U(Q_{*}) = \begin{cases}
               \frac{P_{0} + (1 - P_{1})}{1 - Q_{*}} & \text{if}\; Q_{*} \leq 1 - P_{0} \\
               \frac{1- P_{0}}{Q_{*}} + \frac{1 - P_{1}}{1 - Q_{*}} & \text{if}\; Q_{*} \in [1 - P_{0}, P_{1}] \\
               \frac{P_{1} + (1 - P_{0})}{Q_{*}} &\text{if}\; Q_{*} \geq P_{1}
             \end{cases}             
\]

Notice that the numerators of all these functions are positive, so the first bounding function $(P_{0} + 1- P_{1})/(1 - Q_{*})$ is montonically increasing over its range and the third, $(P_{1} + (1 - P_{0})) / Q_{*}$ is monotonically decreasing over its range. Finally, inspection of the middle bounding function shows that it is convex over its range. This implies that the this function must have its maximum at one of the bound points, $1 - P_{0}$ or $P_{1}$. Taking the maximum of these two values, and comparing them to the maximum two values from the situation when $P_{1} \leq 1 - P_{0}$ gives the expression of the bounds in the result. A similar approach applies to the lower bounds as well.

Sharpness of these bounds is implied by the linear nature of the optimization function and the convexity of the feasible set. If these bounds were not sharp, this would imply that that there are bounds sharper than these that contain all values of $\delta$ consistent with the data and maintained assumptions. But this is clearly contradicted by the fact that the solutions in Table~\ref{t:bfs} are feasible and would fall outside these supposedly sharper bounds. 
\end{proof}

\subsection{Bounds under additional assumptions}\label{sec:bounds-proof}

To derive bounds under additional assumptions, we first derive bounds conditional on the strata probabilities. 

\begin{lemma}\label{lem:general-bounds}
  The bounds on $\delta$ for given values of $\bm{\rho}$ are $\delta \in [\delta_{L}(\bm{\rho}), \delta_{U}(\bm{\rho})]$, where
  \[
  \begin{aligned}    
    \delta_U(\bm{\rho}) &=
    P_{111}\left( \frac{Q_{11}}{Q_{*}} \right) - P_{011}\left( \frac{Q_{01}}{Q_*} \right) - P_{110}\left(\frac{1 - Q_{11}}{1-Q_{*}} \right) + P_{010}\left( \frac{1 - Q_{01}}{1 - Q_*} \right) \\
&+ \frac{1}{Q_*(1 - Q_*)}\min\left\{ 
\begin{matrix}
  P_{110}(1 - Q_{11}) \\
  \rho_{011} + \rho_{001} \\
  Q_* - P_{111}Q_{11}
\end{matrix} \right\} +
\frac{1}{Q_*(1 - Q_*)}\min\left\{
  \begin{matrix}
    P_{011}Q_{01} \\
    \rho_{110} + \rho_{010}\\
    1- Q_* - P_{010}(1 - Q_{01})
  \end{matrix}
\right\},
\end{aligned}
\]
and,
\[
  \begin{aligned}    
    \delta_L(\bm{\rho}) &=
    P_{111}\left( \frac{Q_{11}}{Q_{*}} \right) - P_{011}\left( \frac{Q_{01}}{Q_*} \right) - P_{110}\left(\frac{1 - Q_{11}}{1-Q_{*}} \right) + P_{010}\left( \frac{1 - Q_{01}}{1 - Q_*} \right) \\
&+ \frac{1}{Q_*(1 - Q_*)}\max\left\{ 
\begin{matrix}
  -P_{111}Q_{11} \\
  -\rho_{110} - \rho_{100} \\
  -1 + Q_* + P_{110}(1 - Q_{11})
\end{matrix} \right\} +
\frac{1}{Q_*(1 - Q_*)}\max\left\{
  \begin{matrix}
    -P_{011}Q_{01} \\
    -\rho_{001} - \rho_{101}\\
    -Q_* + P_{011}Q_{01}
  \end{matrix}
\right\}.
\end{aligned}
\]

\end{lemma}

\begin{proof}
Conditional on $\rho_{s}$  and $Q_{*}$, deriving the bounds on $\delta$ is a standard linear programming problem. We now describe the process for deriving these bounds at a general level. Without any assumptions,  we are interested in maximizing or minimizing the objective function,
\[
  \begin{aligned}
    \delta = & P_{111}\left( \frac{Q_{11}}{Q_{*}} \right) - P_{011}\left( \frac{Q_{01}}{Q_*} \right) - P_{110}\left(\frac{1 - Q_{11}}{1-Q_{*}} \right) + P_{010}\left( \frac{1 - Q_{01}}{1 - Q_*} \right) \\
    & + \mu_{011}(1,1)\frac{\rho_{011}}{Q_*(1 - Q_*)} + \mu_{001}(1,1)\frac{\rho_{001}}{Q_*(1 - Q_*)} \\
    & - \mu_{110}(1,1)\frac{\rho_{110}}{Q_*(1 - Q_*)} -\mu_{100}(1,1)\frac{\rho_{100}}{Q_*(1 - Q_*)} \\
    & + \mu_{110}(0,1)\frac{\rho_{110}}{Q_*(1 - Q_*)} + \mu_{010}(0,1)\frac{\rho_{010}}{Q_*(1 - Q_*)} \\
    & - \mu_{101}(0,1)\frac{\rho_{101}}{Q_*(1 - Q_*)} - \mu_{001}(0,1)\frac{\rho_{001}}{Q_*(1 - Q_*)},
  \end{aligned}
\]
subject to the constraints
\[
  \begin{aligned}
    P_{111}Q_{11} &= \mu_{111}(1,1)\rho_{111} + \mu_{101}(1,1)\rho_{101} + \mu_{110}(1,1)\rho_{110} + \mu_{100}(1,1)\rho_{100} \\
    P_{011}Q_{01} &= \mu_{111}(0,1)\rho_{111} + \mu_{011}(0,1)\rho_{011} + \mu_{110}(0,1)\rho_{110} + \mu_{010}(0,1)\rho_{010} \\
    P_{110}(1-Q_{11}) &= \mu_{011}(1,1)\rho_{011} + \mu_{110}(1,1)\rho_{110} + \mu_{010}(1,1)\rho_{010} + \mu_{000}(1,1)\rho_{000} \\
    P_{010}(1-Q_{01}) &= \mu_{101}(0,1)\rho_{101} +
    \mu_{001}(0,1)\rho_{001} + \mu_{100}(0,1)\rho_{100} +
    \mu_{000}(0,1)\rho_{000} \\
    0 & \leq \mu_{s}(t) \leq 1, \quad \forall s, t.
  \end{aligned} 
\]
For this step, we do not need to specify constraints on $\rho_s$ because we consider them fixed (and $Q_*$ is a linear function of $\rho_s$). The simplex tableau method yields the given bounds. 
\end{proof}

\begin{proof}[Proof of Proposition~\ref{prop:monotonicity}]
Recall the constraints on the strata probabilities:
  \[
  \begin{aligned}
    Q_{11} &= \rho_{111} + \rho_{101} + \rho_{110} + \rho_{100} \\
    Q_{01} &= \rho_{111} + \rho_{011} + \rho_{110} + \rho_{010} \\
    Q_* & = \rho_{111} + \rho_{011} + \rho_{101} + \rho_{001},
  \end{aligned}
\]
Under monotonicity, we only have strata $S_{i} \in \{111, 110, 010, 100, 000\}$, so we have  $Q_{*}= \P[M^{*}_i = 1] = \rho_{111}$ and $\rho_{110} + \rho_{010} = Q_{01} - Q_{*}$ and $\rho_{110} + \rho_{100} = Q_{11} - Q_{*}$. Plugging these values into the bounds from Lemma~\ref{lem:general-bounds}, we obtain
\[
\begin{aligned}    
    \delta_U(Q_*) &=
    P_{111}\left( \frac{Q_{11}}{Q_{*}} \right) - P_{011}\left( \frac{Q_{01}}{Q_*} \right) - P_{110}\left(\frac{1 - Q_{11}}{1-Q_{*}} \right) + P_{010}\left( \frac{1 - Q_{01}}{1 - Q_*} \right) \\
&+ \frac{1}{Q_*(1 - Q_*)}\min\left\{ 
\begin{matrix}
  P_{110}(1 - Q_{11}) \\
  0 \\
  Q_* - P_{111}Q_{11}
\end{matrix} \right\} +
\frac{1}{Q_*(1 - Q_*)}\min\left\{
  \begin{matrix}
    P_{011}Q_{01} \\
    Q_{01} - Q_{*}\\
    1- Q_* - P_{010}(1 - Q_{01})
  \end{matrix}
\right\},
\end{aligned}
\]
and,
\[
  \begin{aligned}    
    \delta_L(Q_{*}) &=
    P_{111}\left( \frac{Q_{11}}{Q_{*}} \right) - P_{011}\left( \frac{Q_{01}}{Q_*} \right) - P_{110}\left(\frac{1 - Q_{11}}{1-Q_{*}} \right) + P_{010}\left( \frac{1 - Q_{01}}{1 - Q_*} \right) \\
&+ \frac{1}{Q_*(1 - Q_*)}\max\left\{ 
\begin{matrix}
  -P_{111}Q_{11} \\
  -Q_{11} - Q_{*} \\
  -1 + Q_* + P_{110}(1 - Q_{11})
\end{matrix} \right\} +
\frac{1}{Q_*(1 - Q_*)}\max\left\{
  \begin{matrix}
    -P_{011}Q_{01} \\
    0\\
    -Q_* + P_{011}Q_{01}
  \end{matrix}
\right\}.
\end{aligned}
\]

We further simplify the upper bound expression by noting that $P_{110}(1 - Q_{11}) \geq 0$ and $Q_{01} - Q_{*} \leq 1 - Q_{*} - P_{010}(1 - Q_{01})$. The lower bound simplifies because $Q_{11} - Q_{*} \leq 1 - Q_{*} - P_{110}(1 - Q_{11})$ and $P_{011}Q_{01} \geq 0$. Removing these extraenous conditions gives the result in the text. 

\end{proof}

\begin{proof}[Proof of Proposition~\ref{prop:stability}]
  Under the maintained assumptions, $Q_{*} = Q_{01}$, which we
  plug into the expression of
  Proposition~\ref{prop:monotonicity}. Then, the result is immediate upon noting that $P_{011}Q_{01} - Q_{01} \leq 0$, $P_{011}Q_{01} \geq 0$ and rearranging terms. 
\end{proof}

For calculating the bounds under the sensitivity constraints, we can
take the bounds from Lemma~\ref{lem:general-bounds} and solve a corresponding
linear programming problem to optimize them with respect to the principal strata probabilities. For example, depending the observed data, the upper bound will depend on $\rho_{011} + \rho_{001}$, $\rho_{110} + \rho_{010}$, or $\rho_{011} + \rho_{001} + \rho_{110} + \rho_{010}$. To find the upper bound across values of $\bm{\rho}$, we apply the linear programming machinery to finding the upper bound for each of these quantities subject to the constraints that
\[
  \begin{aligned}
    Q_{11} &= \rho_{111} + \rho_{101} + \rho_{110} + \rho_{100} \\
    Q_{01} &= \rho_{111} + \rho_{011} + \rho_{110} + \rho_{010} \\
    Q_* & = \rho_{111} + \rho_{011} + \rho_{101} + \rho_{001},
  \end{aligned}
\]
where $0 \leq \rho_s \leq 1$ for all $s$ and $\sum_{s\in\mathcal{S}} \rho_s = 1$. Note that for the sensitivity analysis, we may impose additional constraints on $\rho_s$ in this step. As an example, for the objection function of $\rho_{011} + \rho_{001}$, we have the upper bound
\[
  \min\left\{1 - Q_{11}, Q_*, 1 - Q_{01} + Q_*, \frac{1}{2}(1 - Q_{11} + Q_*), 1 + Q_{01} - Q_{11}\right\}.
\]
Plugging these bounds into the upper bound $\delta_U(\bm{\rho})$ will
yield an upper bound purely as a function of observed parameters and $Q_*$ and  $\gamma$ ( the sensitivity parameter). Under some of our assumptions, inspection of the resulting functions reveals that the maximum of these functions can only occur at a handful of critical values of $Q_*$ which can be evaluated and compared quickly. Otherwise, we use a standard optimization routine to find the value of $Q_*$ that maximizes the upper bound or minimizes the lower bound.


\section{Estimation and Inference Details}
\label{sec:estimation_details}

The discussion of the bounds in the main text focused on population level bounds---that is, we identified the bounds in terms of population quantities such as $P_{tzm}$. Estimation and inference for the bounds with a sample poses some important difficulties. The most obvious way to estimate these bounds is to plug in sample version of the population quantities and solve the above linear programming problem to obtain bounds. Unfortunately, the asymptotic distribution of estimators based on this plug-in approach do not have the standard asymptotics due to the lack of differentiability of the bounds as a function of the data. \citet{AndHan09} show that naive bootstrap methods are not valid for these types of estimators due to these issues.

Below we mostly focus on the random placement design, where estimation and inference is the most complicated. In the pre-test design, the bounds are simple functions of the parameters of interest, so we can use standard asymptotic variance estimators, combined with the approach of \cite{ImbMan04} to obtain confidence intervals for the parameters of interest. For the post-test bounds, we obtain standard errors for the bounds based on the nonparametric bootstrap and then use these in the approach of \cite{ImbMan04}. In simulations, we found this to have similar performance to the more complicated \emph{union bound} approach of \cite{YeKeeHas23}.

\subsection{Estimation and inference for the random placement design}

For the random placement design, we do not generally have closed-form expressions for the bounds, and so we can reformulate the estimation and inference problem based on moment conditions that feed into sample versions of a population criterion function \citep{CheHonTam07}. In the general analysis of the random placement design, we define the parameters of our model as
$$
\psi_{y_{1}y_{0},m_{1}}(t, m_{0}) = \Pr[Y^{*}_{i}(t) = y_{1}, Y_{i}(t) = y_{0}, M_{i}(t) = m_{1} \mid M_{i}^{*} = m_{0}],
$$
so that the parameter of interest can be written as
$$
\delta = \sum_{y_{0}m_{1}} \psi_{1y_{0}m_{1}}(1, 1) - \psi_{1y_{0}m_{1}}(0, 1)  - \psi_{1y_{0}m_{1}}(1, 0) + \psi_{1y_{0}m_{1}}(0, 0), 
$$
with constraints
$$
\begin{aligned}
  0 &\leq \psi_{y_{1}y_{0},m_{1}}(t, m_{0}) \leq 1 \\
  &\sum_{y_{1} = 0}^{1} \sum_{y_{0} = 0}^{1} \sum_{m_{1} = 0} ^{1} \psi_{y_{1}y_{0},m_{1}}(t, m_{0}) = 1.
\end{aligned}
$$

Let $W_{i} = (Y_{i}, M_{i}, T_{i}, Z_{i})$ be the observed data vector with possible realized value $w \equiv (w_{y}, w_{m}, w_{t}, w_{z})$. Abusing notation, we let $w_{1} = (w_{y}, w_{m}, w_{t}, 1)$ and $w_{0} = (w_{y}, w_{m}, w_{t}, 0)$. The randomization and consistency assumptions imply the following moment conditions:
$$
\begin{aligned}
  \E[g_{w_{1}}(W_{i}, \psi) \mid Z_{i} = 1] &= \P(Y_{i} = w_{y}, M_{i} = w_{m}, T_{i} = w_{t} \mid Z_{i} = 1) \\ &\qquad - Q_{*} \left(\sum_{y_{0}} \psi_{w_{y}y_{0}w_{m}} (w_{t}, 1)\right) - (1 - Q_{*})\left(\sum_{y_{0}} \psi_{w_{y}y_{0}w_{m}} (w_{t}, 0)\right)  \\
\E[g_{w_{0}}(W_{i}, \psi) \mid Z_{i} = 0] &= \P(Y_{i} = w_{y}, T_{i} = w_{t} \mid M_{i} = w_{m}, Z_{i} = 0)  - \left(\sum_{y_{1}} \sum_{m_{1}} \psi_{y_{1}w_{y}m_{1}} (w_{t}, w_{m})\right).
\end{aligned}
$$
There are $d_{1} = 8$ for the post-test data and $d_{0} = 8$ restrictions for the pre-test data.

Define $r_{d}(\psi)$ encode the deterministic restrictions on the $\psi$ values and let $\Psi^{\dagger}_{d} = \{\psi \;:\; r_{d}(\psi) \geq 0\}$ be the values of the underlying parameters that satisfy these restrictions. These restrictions include assumptions like monotonicity that would cause certain $\psi$ values to be set to zero or the sensitivity analysis specifications that limit the size of a group of $\psi$ values. Under these maintained assumptions, we can characterize the identified set as
$$
\Psi^{\star} = \{\psi \in \Psi^{\dagger}_{d} \;:\; \E[g_{w_{1}}(W_{i}, \psi)] = 0, E[g_{w_{0}}(W_{i}, \psi)] = 0 \;\forall w_{1},w_{0} \in \{0,1\}^{3}\}.
$$

One way to define a distance from the identified set is with a population criterion function. Let the moment conditions be indexed by $j$ such that $j=1,\ldots, 8$ correspond to $\{g_{w_{1}}\}$ and $j=9, \ldots, 16$ correspond to $\{g_{w_{0}}\}$. Then, the population criterion (loss) function is  
\begin{equation}
  \label{eq:3}
  L(\psi) = \sum_{j=1}^{16} |\E[g_{j}(W_{i}, \psi)]|
\end{equation}
Following~\cite{Torgovitsky19}, we use absolute value loss here to ensure that we can leverage linear programming techniques for computational convenience. We can obtain an empirical version of the criterion,
\begin{equation}
  \label{eq:4}
  L_{n}(\psi) = \sum_{j=1}^{16} \sqrt{n}|\overline{g}_{j}(W_{i}, \psi)|,
\end{equation}
where, for instance, 
$$
\begin{aligned}
  \overline{g}_{w_{1}}(W_{i}, \psi) &= \P_{n}(Y_{i} = w_{y}, M_{i} = w_{m}, T_{i} = w_{t} \mid Z_{i} = 1) \\ &\qquad - \widehat{Q}_{*} \left(\sum_{y_{0}} \psi_{w_{y}y_{0}w_{m}} (w_{t}, 1)\right) - (1 - \widehat{Q}_{*})\left(\sum_{y_{0}} \psi_{w_{y}y_{0}w_{m}} (w_{t}, 0)\right),
\end{aligned}
$$
and $\P_{n}$ is the in-sample distribution.

We could proceed with estimating the bounds for a given set of assumptions by searching over the parameter space such that $L_{n} = 0$, but this is often a fragile approach. In particular, it may be the case that restrictions hold in the population but fail to hold in empirical samples due to sampling variability so that the minimum value of $L_{n}$ is strictly greater than 0. As an alternative approach, we can first find the minimum value of $L_{n}$ in the sample and then find extreme values of the parameter under parameter values that are close to that minimum.

We first define the sample minimum of the criterion function under the maintained deterministic restrictions,
\begin{equation}
  \label{eq:5}
  \overline{L}_{n} = \inf_{\psi \in \Psi^{\dagger}_{d}} L_{n}(\psi).
\end{equation}
We can then estimate the upper and lower bounds by finding the minimum and maximum values of $\delta$ that come close to this value:
$$
\begin{aligned}
  \widehat{\delta}_{L} &= \min_{\psi \in \Psi^{\dagger}_{d}} \delta(\psi) \;\text{s.t.}\; L_{n}(\psi) \leq \overline{L}_{n}(1 + \epsilon_{n}), \\
  \widehat{\delta}_{U} &= \max_{\psi \in \Psi^{\dagger}_{d}} \delta(\psi) \;\text{s.t.}\; L_{n}(\psi) \leq \overline{L}_{n}(1 + \epsilon_{n}).
\end{aligned}
$$
The tuning parameter $\epsilon_{n}$ controls how close we require the criterion function of the bounds to be to the overall sample minimizer. This approach requires $\epsilon_{n} \to 0$ as $n \to \infty$. We take $\epsilon_{n} = 0.25$ in our implementation, which has shown good performance in simulations.

For the confidence intervals, we use the nonparametric bootstrap to obtain standard error estimates of each of the bounds and then apply the \citep{ImbMan04} approach. In simulations, we found this approach to be slightly conservative and other competing methods, such as \cite{CheNewSan23}, to undercover the true parameter slightly, at least in simulations similar to our data example. Thus, we use the nonparametric bootstrap approach for our confidence intervals. 

\section{Parametric Bayesian Approach to Incorporate Covariates}\label{covariates}

The nonparametric bounds above are sharp in the sense that they leverage all
information about the outcome, moderator, treatment, and question order.
Researchers, however, often have additional data in the form of covariates that
may help reduce the uncertainty of their estimates. Here, we consider a Bayesian
parametric model of the principal strata approach to the pre-test, post-test,
and random placement designs, building on the work of \citet{MeaPac13}
\citep[see also][]{ImbRub97,HirImbRub00}. Unlike the nonparametric bounds
approach, a Bayesian model allows us to incorporate prior information about the
data-generating process in a smooth and flexible
manner.\footnote{\cite{LevBonZen23} proposes a way to incorporate covariates on
  nonparametric bounds when the quantity of interest can be written as an
  average of covariate-specific quantities. Unfortunately, we cannot write the
  interaction in this way because it is the difference between two different
  CATEs that condition on different subsets of the data. One could use their
  approach on each of the individual CATEs and combine those bounds for the
  interaction, but the resulting bounds would not be
  sharp.} 

\subsection{The Model}

Our approach focuses on a data augmentation strategy that models the
joint distribution of the outcomes and the principal strata, the
latter of which are not directly observable. We allow the distribution
of the potential outcomes and principal strata conditional on those
strata to further depend on covariates via a binomial and multinomial
logistic model, respectively:
\[
  \begin{aligned}
    \P(Y_i= 1 \mid T_i = t, Z_i = z, S_i = s, \bm{X}_i) &= \mu_{is}(t, z) = \textrm{logit}^{-1}(\alpha_{tz|s} + \mathbold{X}'_i\mathbold{\beta}), \\
    \P(S_i = s \mid \mathbold{X}_i) &=\rho_{is} =  \frac{\exp(\mathbold{X}'_i\mathbold{\psi}_s)}{\sum_{j \in \mathcal{S}} \exp(\mathbold{X}'_i\mathbold{\psi}_j)},
  \end{aligned}
\]
where $\mathbold{X}_i$ denotes observed pre-treatment covariates that might be predictive of unit $i$'s outcome and principal strata. Note that the strata probabilities do not depend on $T_i$ and $Z_i$ due to randomization. We gather the parameters as $\mathbold{\alpha} = \{\alpha_{tz|s}\}$ and $\mathbold{\psi} = \{\mathbold{\psi}_s\}$.
We can easily incorporate  assumptions like monotonicity and stable moderators by simply restricting the space of possible principal strata $\mathcal{S}$.

Our goal is to make inferences about the posterior distribution of these parameters and the ultimate quantities of interest, $\tau(d)$ and $\delta$. There are two ways to represent these quantities under this parametric model, resulting in two different types of posterior distributions. The first is based on \emph{population} inference and derives expressions for $\tau(d)$ and $\delta$ purely in terms of the parameters of the model. The second is based on \emph{in-sample} inference and derives expressions for $\tau(d)$ and $\delta$ in terms of  potential outcomes in a particular sample.

For the population inference approach, we first note that due to consistency and randomization, we have $\mu_{is}(t, z) = \P(Y_i(t, z) = 1 \mid S_i = s, \mathbold{X}_i)$. Thus, we can write the values of the quantities of interest for a given unit as,
\[
   \tau_i(m) = \E\left[  Y^{*}_i(1) - Y^{*}_i(0)  \mid M^{*}_{i} = m, X_i \right] \ =  \ \sum_{s \in \mathcal{S}^*_d} (\mu_{is}(1,1)- \mu_{is}(0, 1))\rho_{is},
\]
and $\delta_i = \tau_i(1) - \tau_i(0)$, where we omit the implied dependence on
$(\mathbold{\alpha, \beta, \psi})$ and remember that $\mathcal{S}^*_d$
is the set of strata levels such that $M^{*}_{i} = m$, for $m\in \{0,1\}$.  Using the
empirical distribution of covariates, the average of these conditional mean differences and
interactions will equal the overall quantities of interest, i.e., $\tau(m) = \sum_{i=1}^n \tau_{i}(m)/n$ and 
$\delta = \sum_{i=1}^n \delta_i/n$.

The in-sample versions of the quantities of interest are more straightforward,
since they are simply the conditional mean differences and interaction among the units in the sample,
\[
\tau_s(m) = \frac{\sum_{i=1}^n \mathbb{I}(M^{*}_{i} = m)\left\{ Y^{*}_i(1) - Y^{*}_i(0) \right\}}{\sum_{i=1}^n \mathbb{I}(M^{*}_{i} = m)},
\]
and $\delta_s = \tau_s(1)-\tau_s(0)$. Obviously, across repeated samples, we can relate these to the population quantities as $\E[\tau_s(m)] = \tau(m)$ and $\E[\delta_s] = \delta$. 

We develop an efficient Markov Chain Monte Carlo (MCMC) algorithm to take draws from the posterior and then calculate these quantities of interest. Our Gibbs sampler can also be simplified and used for inference in the absence of pre-treatment covariates, which can be viewed as a Bayesian alternative to uncertainty estimation for the partially identified parameters discussed in Section~\ref{subsection_statistical_inference}. We provided details of these algorithms in Supplemental Materials~\ref{section_mcmc}.

This Bayesian approach has the advantage of easily incorporating
covariates, but it does require us to select prior distribution for
the model parameters, some of which are unidentified in the
frequentist sense. Thus, the identification of these parameters will
depend on the prior. To investigate this, we take draws
of the prior predictive distribution under different prior structures,
which we show in Figure~\ref{fig:ppd}. All the priors we considered
are symmetric, but uniform priors on the model parameters lead to
somewhat informative priors on the ultimate quantity of
interest. Thus, we rely on
more dispersed priors for the simulations and the application. We discuss the choice of prior
distribution more fully in Supplemental Materials~\ref{section_mcmc}.

We conduct two simulation studies to demonstrate the gains in efficiency from the monotonicity and stability assumptions and the
incorporation of covariates in the Bayesian approach. The first simulation varies the assumptions of the data-generating process and compares the posterior variance of these distributions across combinations
of our two assumptions. The second simulation varies the predictive
power of the covariates on the outcome and the strata in the data-generating process and compares the variance of the posterior
distributions from Gibbs run on each simulated data set with and
without incorporating covariates. We present these results in Supplemental Materials ~\ref{section_sim}.

\section{MCMC Algorithm}
\label{section_mcmc}

In this section we describe our MCMC algorithm for the Bayesian model of Section~\ref{covariates}. Our goal is to sample from the joint distribution of the parameters and the principal strata indicator,
\[
  \begin{aligned}
    \P(\mathbold{\alpha, \beta, \psi}, \mathbold{S} &\mid \mathbold{Y}, \mathbold{X, T, Z, M})  \propto \\
    \prod_{i=1}^n &\left(\sum_{s \in \mathcal{S}_i} \left[ \P(Y_i\mid T_i, Z_i, S_i = s, \bm{X}_i)\P(S_i = s \mid \mathbold{X}_i)\right]^{\mathbb{I}(S_i = s)} \right)\P(\mathbold{\alpha, \beta, \psi}),
  \end{aligned}
\]
where $\mathcal{S}_i = \mathcal{S}(T_i, Z_i, M_i)$ are the set of principal strata to which unit $i$ could possibly belong. When the set of observed pre-treatment covariates ($\bm{X}_i$) is empty, the parameter space reduces to that of a standard finite mixture model, and sampling from the joint posterior is straightforward. With $\bm{X}_i$, Bayesian inference for the model is more complicated. Traditionally, Bayesian inference for logistic regression models has been challenging due to a lack of a simple Gibbs sampling algorithm. Recently, however, \citet{PolScoWin13} introduced a simple data-augmentation strategy based on the P\'{o}lya-Gamma (PG) distribution, obviating the need for approximate methods or precise tuning of a Metropolis-Hastings algorithm. We use this approach for both the binary and multinomial logistic regression models for the outcome and principal strata, respectively. This allows a simple Gibbs structure where the full conditional posterior distributions of $(\mathbold{\alpha, \beta})$ and $\mathbold{\psi}$ are Normal conditional on specific draws from the PG distribution. 

Conditional on the other parameters, then the full conditional posterior of the principal strata follows a similar form to \citet{HirImbRub00},
\[
\P(S_i = s \mid Y_i, \mathbold{X}_i, T_i, Z_i, M_i, \mathbold{\alpha, \beta, \psi}) = \frac{\mu_{is}(T_i, Z_i)^{Y_i}( 1- \mu_{is}(T_i, Z_i))^{1-Y_i}\rho_{is}}{\sum_{k \in \mathcal{S}_i} \mu_{ik}(T_i, Z_i)^{Y_i}(1-\mu_{ik}(T_i, Z_i))^{1- Y_i}\rho_{ik}},
\]
where we suppress the dependence of $\mu_{is}$ and $\rho_{is}$ on the model parameters. Repeatedly drawing from these full conditional posterior distributions should provide a sample from the above joint posterior and allow for posterior inference in the usual manner. In each iteration, $r \in \{1, \ldots, R\}$, of the algorithm, we have draws
\[
\left( \left\{ \widehat{S}_i^{(r)} \right\}_{i=1}^n, \widehat{\mathbold{\psi}}^{(r)}, \widehat{\mathbold{\alpha}}^{(r)}, \widehat{\mathbold{\beta}}^{(r)} \right).
\]

We can use these draws to generate draws of the population and in-sample versions of the quantity of interest. Given that $\widehat{S}^{(r)}_i$  is the imputed principal strata imputed for unit $i$ in the $r$th draw from the posterior, we let
\[
\widehat{\mu}^{(r)}_{i}(t,z) = \widehat{\mu}_{i,\widehat{S}^{(r)}_i}(t,z, \widehat{\mathbold{\alpha}}^{(r)}, \widehat{\mathbold{\beta}}^{(r)})
\]
be the mean of the potential outcomes conditional on that imputed principal strata. Furthermore, let $\widehat{\rho}^{(r)}_{is}$ be the $r$th draw of the predicted probabilities of each principal strata for each unit. Then, we can calculate the population quantity as
\[
\widehat{\delta}_p^{(r)} = \frac{1}{n} \sum_{i=1}^n \left(\sum_{s \in \mathcal{S}^*_1} (\widehat{\mu}^{(r)}_{is}(1,1)- \widehat{\mu}^{(r)}_{is}(0, 1))\widehat{\rho}^{(r)}_{is}\right) - \left(\sum_{s \in \mathcal{S}^*_0} (\widehat{\mu}^{(r)}_{is}(1,1)- \widehat{\mu}^{(r)}_{is}(0,1))\widehat{\rho}^{(r)}_{is}\right),
\]

For the in-sample quantity, we can then draw \emph{imputed} values of the missing potential outcomes themselves $\widehat{Y}^{*,(r)}_i(1) \sim \text{Bin}(\widehat{\mu}^{(r)}_{i}(1,1))$ and $\widehat{Y}^{*,(r)}_i(0) \sim \text{Bin}(\widehat{\mu}^{(r)}_{i}(0,1))$. We can combine this with the imputed value of $M^{*}_i$, which mechanically derives from $\widehat{S}_i^{{r}}$, to get the $r$th draw from the posterior of $\delta_s$,
\[
\widehat{\delta}^{(r)}_s =  \frac{\sum_{i=1}^n \widehat{M}^{*,(r)}_{i}\left\{ \widehat{Y}^{*,(r)}_i(1) - \widehat{Y}^{*,(r)}_i(0) \right\}}{\sum_{i=1}^n \widehat{M}^{*,(r)}_{i}} - \frac{\sum_{i=1}^n \left(1 - \widehat{M}^{*,(r)}_{i}\right)\left\{ \widehat{Y}^{*,(r)}_i(1) - \widehat{Y}^{*,(r)}_i(0) \right\}}{\sum_{i=1}^n \left(1 - \widehat{M}^{*,(r)}_{i}\right)}.
\]
Broadly speaking, we would not expect very large differences between these two targets, except for slightly less posterior variance for the in-sample version.

\begin{figure}[htbp]
    \centering
        \caption{Prior predictive distribution of the parameter under three different prior distributions: (red) the default priors that scales a Jeffreys prior by the number of principal strata; (blue) a uniform prior on all parameters; and (green) a more extreme prior that has $\alpha = 1/16$.}
        \includegraphics[width=1\textwidth]{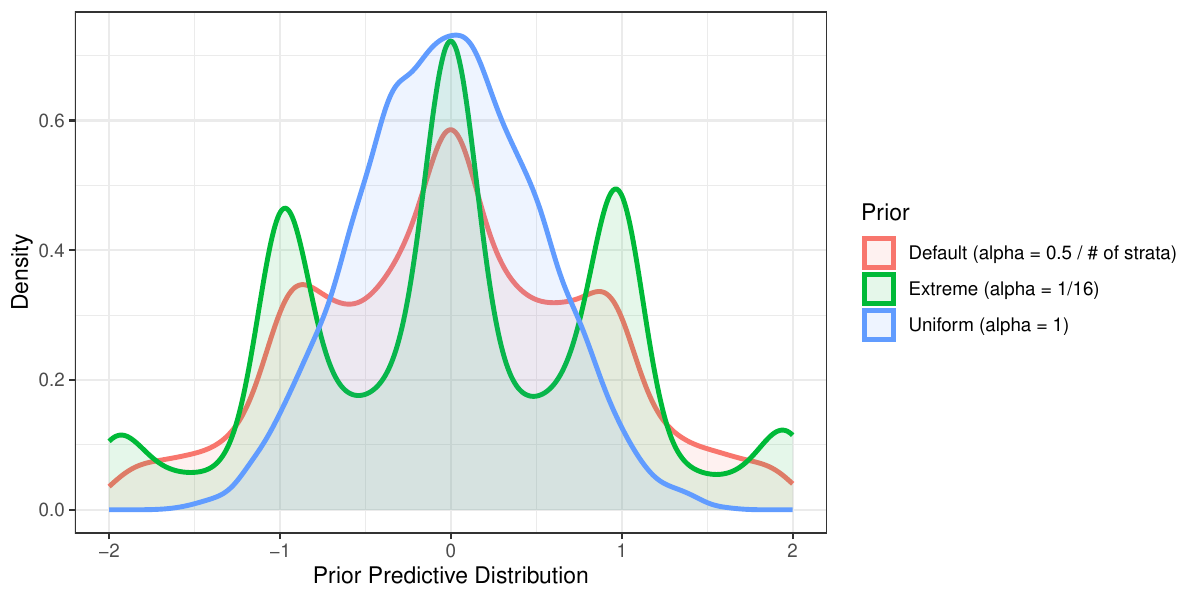}
        \label{fig:ppd}
\end{figure}

As discussed in the main text, the priors need careful attention because they drive the identification of the parameters that are unidentified by the likelihood. One additional complication comes from how the ultimate quantity of interest is a function of the parameters so we cannot directly place, for example, a uniform prior on $\delta$. Figure~\ref{fig:ppd} shows the prior predictive distribution for interaction with three different priors when monotonicity and stable moderator under control are assumed and there are no covariates. The uniform prior on all parameters results in a prior on $\delta$ that has more density in the center of identified range than we might expect. This result is similar to how sums of uniform random variables are not themselves uniform. We can counteract this issue by reducing the Dirichlet and Beta hyperparameters below 1 to put more density at extreme values of the parameters compared to the center. Dropping these parameters down to 1/16 (in green) leads to more mass on strata means closer to 0 or 1 and strata probabilities closer to 0 and 1. In terms of the interaction, this leads to more mass at the values -2, -1, 0, 1, and 2. Our default prior (red) is one that scales the hyperparameters by the inverse of the number of strata to achieve something closer to a uniform distribution.

Additionally, we re-ran the Gibbs empirical analysis of the \cite{horowitz2020can} study, adjusting the priors to the extreme values or uniform values in the previous simulation. The results are displayed in Figure~\ref{fig:hk-priors-plot}, demonstrating the general consistency of the point estimates across starting priors, although there is some fluctuations in the variance of the results. 

\begin{figure}[t!]
    \centering
        \caption{Comparing Bayesian estimates for $\delta$ for default, extreme, and uniform priors}
        \includegraphics[width=1\textwidth]{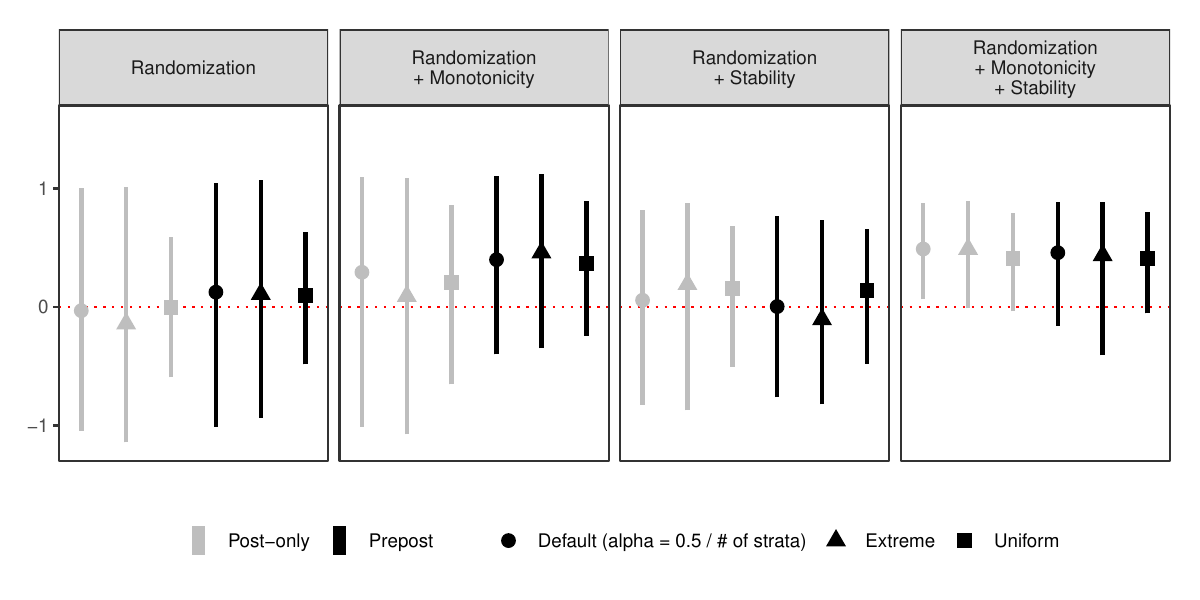}
        \caption*{\textit{Notes:} Figure shows posterior means and 95\% credible intervals for $\delta$ under
        different sets of assumptions, applied to either the post-only data (grey) or the combined pre-post data (black).
        Estimates are shown with default, extreme, and uniform priors, (denoted by circles, triangles, and squares, respectively).
        We follow the original authors in using age, gender, education, and closeness to one's ethnic group as
        covariates. The na\"{i}ve OLS estimates are included for comparison.}
        \label{fig:hk-priors-plot}
\end{figure}

\subsection{Simulation Evidence for the Bayesian Approach}
\label{section_sim}

\paragraph{Simulation Study~I.}
In the first simulation study, we generate simulated data with $n =
1000$ constructed using a data generating process that matches the
Bayesian posterior, pre-specifying coefficient values for the outcome
and principal strata models, randomly drawing values of $Z$, $T$, and
three covariates $X_1$, $X_2$, and $X_3$, and generating values of $Y$
and $M$ from the models. Tables~\ref{tab:betas_assumptions}
and~\ref{tab:psis_assumptions} in the Supplemental Materials display the $\beta$
coefficients for the outcome and $\psi$ coefficients for the for the
true data generating process (DGP). The DGP assumes that monotonicity and stable moderator under
control both hold so that there are three feasible strata
($\mathcal{S} = \{000, 100, 111\}$). Thus, in this setting it would be
most appropriate to incorporate both assumptions into the MCMC
algorithm for sampling from the posterior distribution. Since these
assumptions narrow the nonparametric bounds, we expect the assumptions
to reduce variance of the posterior distribution of $\delta$.

To test this, we perform a Monte Carlo simulation with 1,000
iterations. For each iteration, we calculate the posterior
distribution of $\delta$ with the same data across four different
versions of the MCMC algorithm:  enforcing just the monotonicity
assumption, enforcing just the stable moderator under control
assumption, enforcing neither assumption, and enforcing both
assumptions. Each run of our MCMC algorithm consists of 4 chains with
2,000 iterations each, 200 burn-in (or warm-up) iterations, and a
thinning parameter of 2. Both in-sample and population $\delta$ values
are calculated at each iteration and the variance of the posterior is
calculated from a sample of 1,000 draws from the posterior. This is
done for each of the 1,000 simulated datasets, and for each dataset we
compute the percent reduction in variance compared to the MCMC
algorithm with no assumptions when using the algorithm with the
monotonicity assumption, the stable moderator under control
assumption, or both assumptions.

\begin{figure}[t!]
	\centering
	\includegraphics[width=0.8\textwidth]{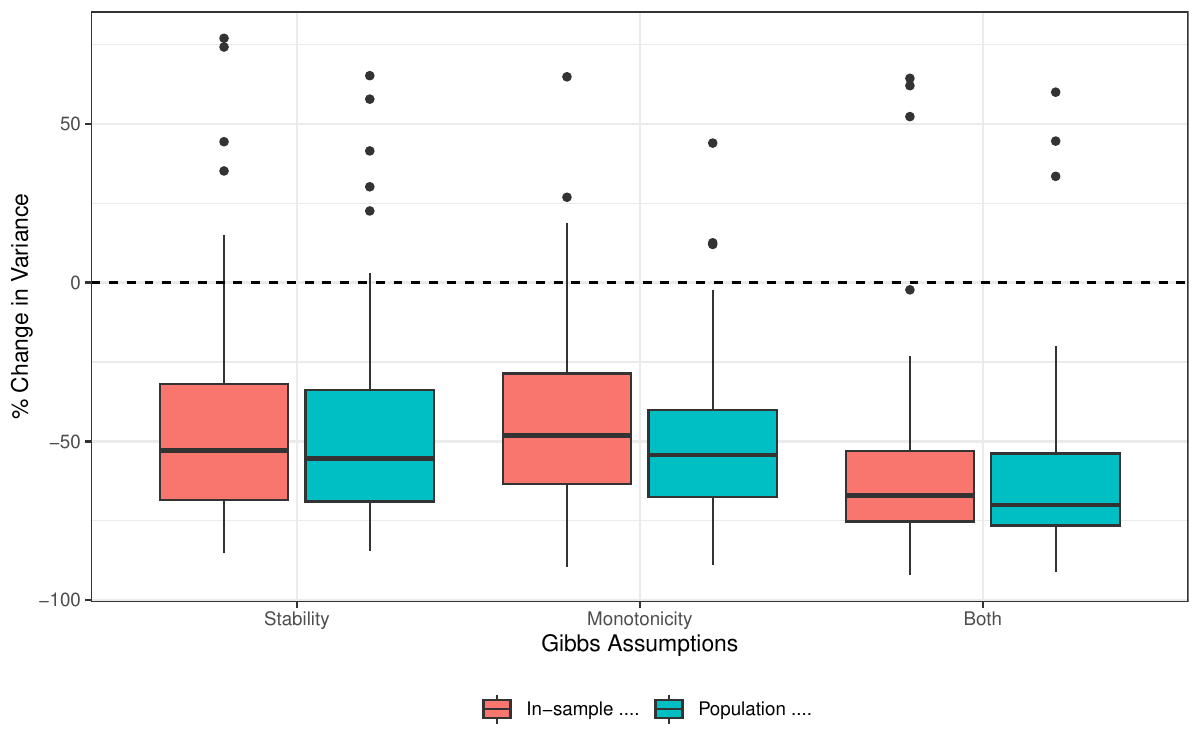}
	\caption{Variance reduction from different combinations of
          assumptions. Boxplots present distribution of \% variance reduction of $\delta$
from the MCMC algorithm with the labeled assumptions compared to same
algorithm with no assumptions, across 1,000 draws of simulated
data. For each simulation iteration, 4 chains were run for each
combination of assumptions. MCMC parameters: 2,000 iterations, 200
burn-in, 2 thinning parameter, simulated data $n = 1,000$.}
	\label{fig:sim_assumptions}
\end{figure}

Figure~\ref{fig:sim_assumptions} presents boxplots for the
distribution of reductions in variance for each combination of
assumptions. Both the monotonicity and stable moderator assumptions on
their own reduce the variance compared to no assumptions, while making
both assumptions reduces the variance even further. The monotonicity
assumption showed a median posterior variance reduction of 40.8\% for
the in-sample $\delta$ and 42.0\% for the population $\delta$. The
stable moderator under control assumption on reduced the posterior
variance by a median reduction of 47.1\% (in-sample $\delta$) and
50.4\% (population $\delta$). The MCMC algorithm with both assumptions
exhibited a posterior variance reduction of 59.3\% (in-sample
$\delta$) and 61.7\% (population $\delta$).

\paragraph{Simulation Study~II.}
In the second simulation study, we drew a series of simulated datasets
under different conditions where the covariates had a weak, medium, or
strong correspondence with the outcome and principal strata in the
data generating processes. Thus, there were six total conditions:
Weak, Medium, and Strong influence in the outcome DGP; and Weak,
Medium, and Strong influence in the principal strata DGP. The values
of the coefficients for these conditions are  $\beta$ and $\psi$
values of 0, 0.25, and 0.5, respectively. When varying the influence
of covariates in the outcome DGP, the influence of covariates on the
strata was held constant, and the influence in the outcome model was
similarly held constant when varying influence in the strata
DGP. Fixed values of the $\beta$'s and $\psi$'s are shown in the
Supplemental Materials in
Tables~\ref{tab:beta_covar_sim}~and~\ref{tab:psi_covar_sim}.

\begin{figure}[t!]
	\centering
	\includegraphics[width=0.8\textwidth]{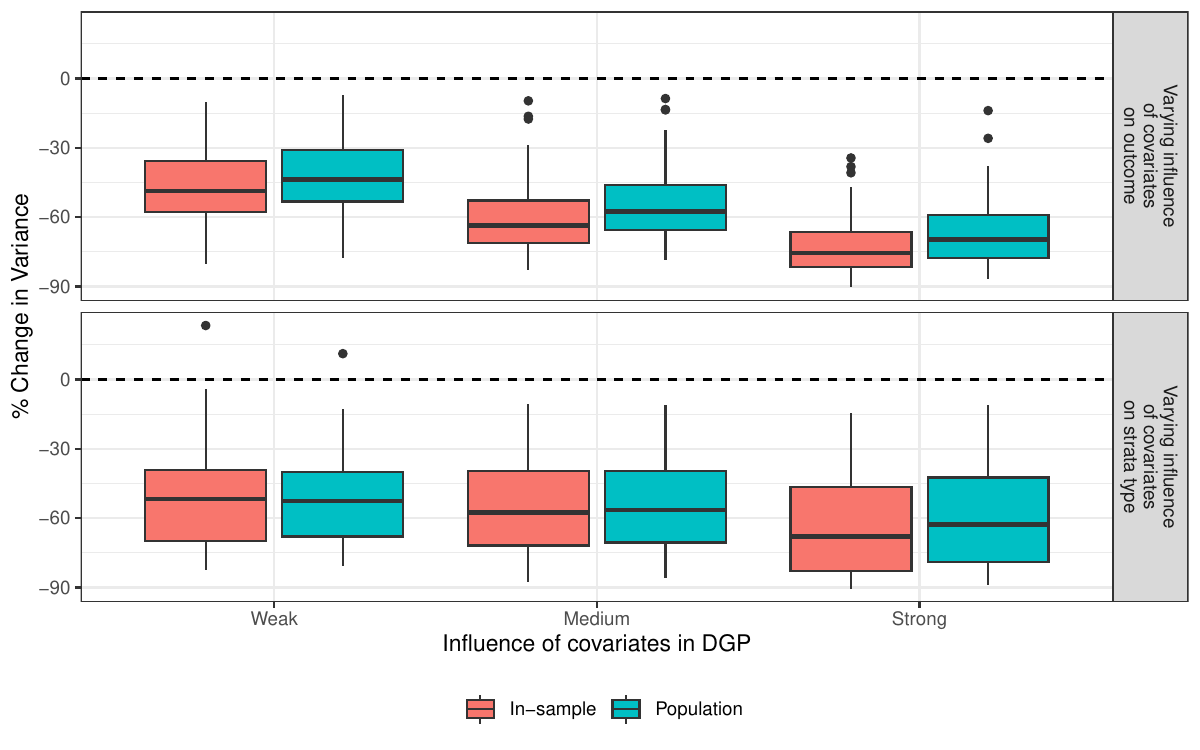}
	\caption{Variance reduction from incorporation of
          covariates. Boxplots present distribution of \% variance
          reduction of $\delta$ from Gibbs with covariates compared to
          MCMC without covariates, across 1,000 draws of simulated
          data. For each draw 4 MCMC chains were run for each
          combination of assumptions. MCMC parameters: 2,000
          iterations, 200 burn-in, 2 thinning parameter, simulated
          data $n = 1,000$.}
	\label{fig:sim_covars}
\end{figure}

For each condition, we drew 1,000 simulated datasets and ran the MCMC
algorithm twice: one time incorporating covariates and one time
omitting them. Each MCMC run consisted of the same iterations,
burn-in, and thinning parameters as in the previous simulation
study. We again calculate in-sample and population $\delta$ values for
each iteration of the Gibbs and calculate the variance of the
posterior distribution and the \% variance reduction comparing the
Gibbs with covariates to that without. Figure~\ref{fig:sim_covars}
presents boxplots for the distribution in variance reduction. When we
vary the influence of covariates on the outcome, we see a clear
variance reduction in all conditions, and we observe a larger reduction as the
influence of covariates on the outcome in the DGP increases. When
testing the impact of incorporating covariates across different levels
of influence in the DGP on the strata, the pattern is less pronounced,
with overall reduction increases in all conditions but slightly lower
reductions in the Medium than Weak condition. The Strong condition
still has the largest variance reduction overall, however, so in
general the efficiency gains from incorporating covariates are
increasing as the influence of covariates on strata in the data
increases.

\section{Additional Simulation Details}

\begin{table}[ht]
\caption{$\beta$ Values for DGP in Bayesian Assumptions Simulation}
\label{tab:betas_assumptions}
\centering
\begin{tabular}{rr}
  \hline
 Variable & $\beta$  \\
  \hline
(Intercept) & -2.00 \\
  X1 & 1.00 \\
  X2 & 0.15 \\
  X3 (Medium) & 0.24 \\
  X3 (Large) & 0.28 \\
  T & 0.83 \\
  Z & -0.01 \\
  T:Z & 0.11 \\
  S111 & 0.41 \\
  S100 & 0.62 \\
  T:S111 & 0.01 \\
  T:S100 & 0.23 \\
  Z:S111 & 0.20 \\
  Z:S100 & -0.02 \\
  T:Z:S111 & -0.90 \\
  T:Z:S100 & 0.09 \\
   \hline
\end{tabular}
\end{table}

\begin{table}[ht]
\caption{$\psi$ Values for DGP in Bayesian Assumptions Simulation}
\label{tab:psis_assumptions}
\centering
\begin{tabular}{rrrr}
  \hline
 & S111 & s100 & s000  \\
  \hline
(Intercept) & -2.06 & -1.00 & 0.00  \\
  X1 & 2.00 & 1.50 & 0.00   \\
  X2 & 0.50 & 0.17 & 0.00   \\
  X3 (Medium) & 1.35 & -0.28 & 0.00  \\
  X3 (Large) & 1.75 & -1.01 & 0.00  \\
   \hline
\end{tabular}
\end{table}

\begin{table}[ht]

\centering
\caption{Fixed $\beta$ Values in Covariate Simulation}
\label{tab:beta_covar_sim}
\begin{tabular}{rr}
  \hline
Variable & $\beta$ \\
  \hline
(Intercept) & -1.00 \\
  X1 & 1.00 \\
  X2 & 0.50 \\
  X3 (Medium) & 0.50 \\
  X3 (Large) & 0.28 \\
  T & 0.83 \\
  Z & -0.01 \\
  T:Z & 0.11 \\
  S111 & 0.41 \\
  S100 & 0.62 \\
  T:S111 & 2.00 \\
  T:S100 & -0.13 \\
  Z:S111 & 0.50 \\
  Z:S100 & 0.10 \\
  T:Z:S111 & 0.05 \\
  T:Z:S100 & 0.01 \\
   \hline
\end{tabular}
\end{table}

\begin{table}[ht]
\centering
\caption{Fixed $\psi$ values in Bayesian Covariate Simulation}
\label{tab:psi_covar_sim}

\begin{tabular}{rrrr}
  \hline
 & S111 & S100 & S000 \\
  \hline
(Intercept) & -2.06 & -1.00 & 0.00 \\
  X1 & 2.00 & 1.50 & 0.00 \\
  X2 & 0.50 & 0.17 & 0.00 \\
  X3 (Medium) & 1.35 & -0.28 & 0.00 \\
  X3 (Large) & 1.75 & -1.01 & 0.00 \\
   \hline
\end{tabular}
\end{table}

\subsection{Comparing the sharp bounds and Bayesian approach without covariates}

\begin{figure}[t!]
    \centering
        \caption{Comparing non-parametric bounds and Bayesian estimates for $\delta$ under different assumptions}
        \includegraphics[width=1\textwidth]{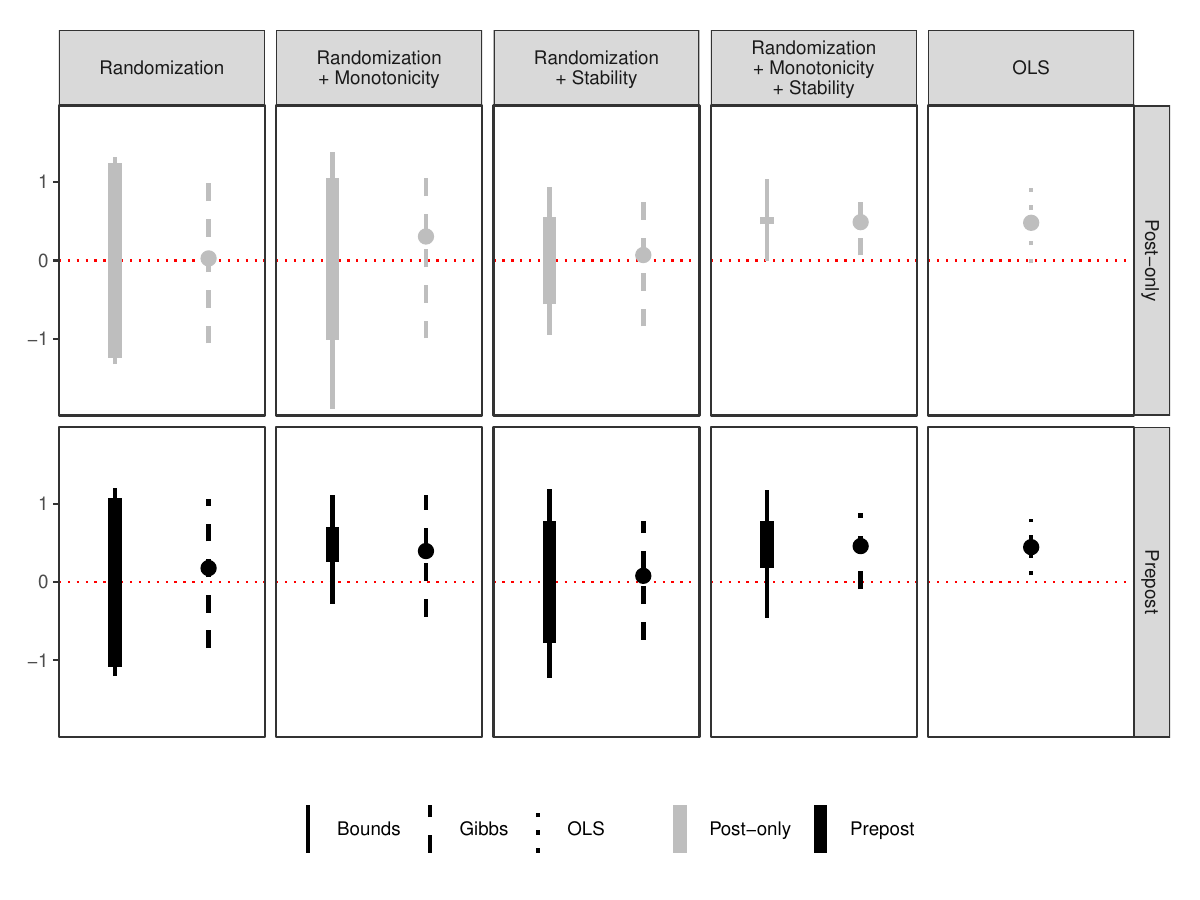}
        \caption*{\textit{Notes:}
        The figure shows nonparametric bounds and Bayesian estimates of the quantity of interest
        under different sets of assumptions applied to either the post-only data (grey) or the combined pre-post data
        (black). The thick bars denote the width of the bounds, and thinner lines denote the 95\% confidence intervals
        around the bounds. Across the first four panels, the thin lines with dots denote the Bayesian posterior mean
        and 95\% credible interval. This estimate included no covariates to facilitate comparison with the
        nonparametric bounds. For the final panel (``OLS''), the thin lines with dots denote the OLS estimate
        and 95\% confidence interval. 
          }
        \label{fig:hk-plot2}
\end{figure}

Figure~\ref{fig:hk-plot2} displays the non-parametric bounds (with 95\% confidence
intervals) and Bayesian estimates (posterior means with 95\% credible intervals)
for $\delta$ under different sets of assumptions applied to the post-only data
(in grey) and to the combined pre-post data (in black). We also include the
na\"{i}ve OLS estimate with 95\% confidence interval for comparison in
the final panel.

\subsection{Incorporating covariates into the Bayesian approach}

\begin{figure}[t!]
    \centering
        \caption{Comparing Bayesian estimates for $\delta$ with and without covariates}
        \includegraphics[width=1\textwidth]{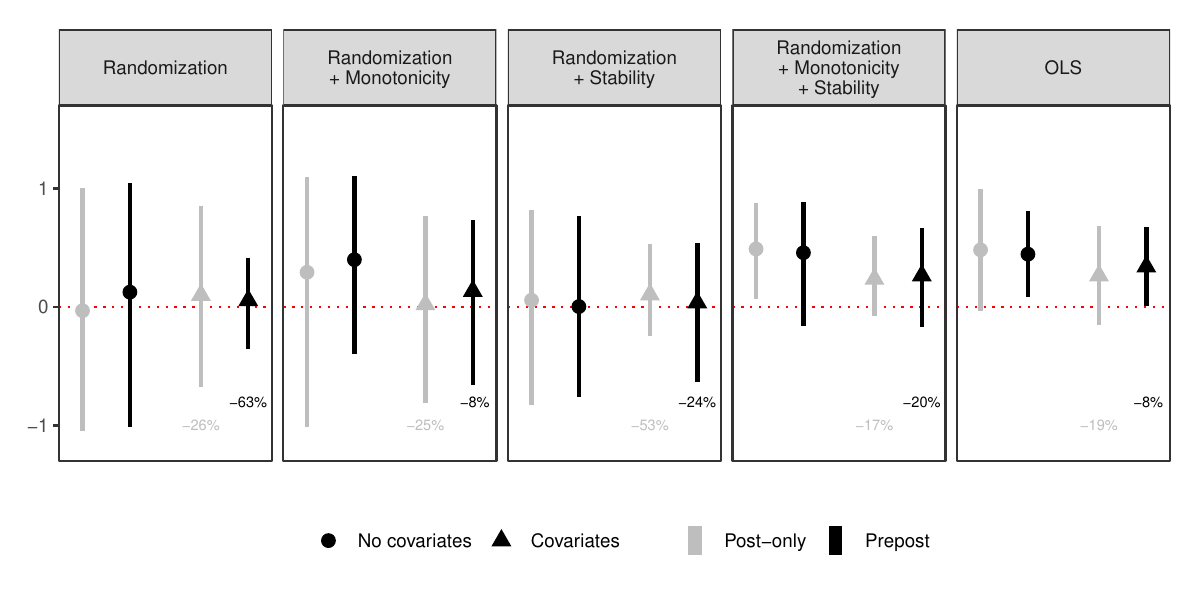}
        \caption*{\textit{Notes:} Figure shows posterior means and 95\% credible intervals for $\delta$ under
        different sets of assumptions applied to either the post-only data (grey) or the combined pre-post data (black).
        Estimates are shown with and without the inclusion of covariates (denoted by triangles and circles, respectively),
        and the numbers indicate the reduction in the width of the credible intervals due to the inclusion of covariates
        for the post-only data (in grey) and the combined pre-post data (in black).
        We follow the original authors in using age, gender, education, and closeness to one's ethnic group as
        covariates. The na\"{i}ve OLS estimates are included for comparison.}
        \label{fig:hk-plot3}
\end{figure}

Thus far,  our Bayesian estimates have omitted covariates to aid comparison with the nonparametric bounds. However, as discussed above, a key
attraction of the Bayesian approach is the ease with which we can
incorporate additional information.  Figure~\ref{fig:hk-plot3}
presents posterior means with 95\% credible intervals, both with and
without covariates, under different assumptions.  We follow
the original authors in using age, gender, education, and closeness to
one's ethnic group as covariates. Including covariates
significantly tightens the credible intervals, especially when fewer
assumptions are imposed.  For example, when only randomization is
assumed, the width of the 95\% credible interval shrinks by more than
50\% when including covariates.  While including covariates
does not alter our substantive conclusions in this case, it does show
that incorporating additional information can lead to large gains in
precision. Since researchers often include a wide range of control
variables in the design of a survey experiment, flexibly leveraging this information is a key advantage of the
Bayesian approach.


\end{appendices}
\fi

\end{document}